\documentclass[sigconf, nonacm, balance=false, pdfa]{acmart}

\newcommand\vldbdoi{XX.XX/XXX.XX}
\newcommand\vldbpages{XXX-XXX}
\newcommand\vldbvolume{19}
\newcommand\vldbissue{13}
\newcommand\vldbyear{2026}
\newcommand\vldbauthors{\authors}
\newcommand\vldbtitle{\shorttitle} 
\newcommand\vldbavailabilityurl{https://github.com/GauvainD/GRAFT}
\newcommand\vldbpagestyle{empty} 

\usepackage[a-2b]{pdfx}
\AtEndPreamble{\hypersetup{pdflang=\relax, pdfdisplaydoctitle=false}}
\usepackage{xspace}
\usepackage{multirow}
\usepackage{amsmath}
\usepackage{amsthm}

\newtheorem*{thm*}{Theorem}

\newtheorem{definition}{Definition}
\newtheorem*{definition*}{Definition}
\newtheorem{lemma}{Lemma}

\newtheorem{example}{Example}

\newtheorem{theorem}{Theorem}
\usepackage{paralist}
\usepackage{listings}
\usepackage{soul}
\usepackage{graphicx}
\usepackage[ruled,linesnumbered]{algorithm2e}
\usepackage{caption}
\usepackage{subcaption}
\usepackage{tabularx}
\usepackage{pbalance}
\newtheorem{innercustomthm}{Theorem}
\newenvironment{extthm}[1]
  {\renewcommand\theinnercustomthm{#1}\innercustomthm}
  {\endinnercustomthm}
\newtheorem{innercustomdef}{Definition}
\newenvironment{extdef}[1]
  {\renewcommand\theinnercustomdef{#1}\innercustomdef}
  {\endinnercustomdef}

\newcommand{\tset}{\ensuremath{\Sigma}} 
\newcommand{\powset}[1]{\ensuremath{\mathcal{P}(#1)}} 
\newcommand{\schemaset}{\ensuremath{\mathcal{S}}\xspace}
\newcommand{\transfoset}{\ensuremath{\mathcal{M}}\xspace}
\newcommand{\target}{\ensuremath{T}\xspace}
\newcommand{\simtarget}{\ensuremath{T'}\xspace}
\newcommand{\schema}[1]{\ensuremath{S_{#1}}\xspace}
\newcommand{\aschema}{\ensuremath{S}\xspace}

\newcommand{\transfo}[1]{\ensuremath{M_{#1}}\xspace}
\newcommand{\atransfo}{\transfo{}}

\newcommand{\simil}[2]{\ensuremath{sim(#1,#2)}}

\newcommand{\propgraph}{\ensuremath{G}\xspace}
\newcommand{\editset}{\ensuremath{E}\xspace}
\newcommand{\meditset}{\ensuremath{E^{*}}\xspace}
\newcommand{\graft}{GRAFT\xspace}
\newcommand{\beam}{\ensuremath{\beta}\xspace}

\newcommand{\reach}{\ensuremath{\mathit{Reach}}}
\newcommand{\besttarget}{\ensuremath{T_b}}

\newcommand{\nextreln}{\ensuremath{Next}\xspace}
\newcommand{\startreln}{\ensuremath{Start}\xspace}
\newcommand{\mnextreln}{\ensuremath{\nextrel^{*}}\xspace}
\newcommand{\mstartreln}{\ensuremath{\startrel^{*}}\xspace}

\newcommand{\AddLabeln}{AddLbl\xspace}
\newcommand{\RmvLabeln}{RmvLbl\xspace}
\newcommand{\AddObjectn}{Add\xspace}
\newcommand{\RmvObjectn}{Rem\xspace}
\newcommand{\AddEdgen}{AddE\xspace}
\newcommand{\AddPropn}{AddProp\xspace}
\newcommand{\RmvPropn}{RmvProp\xspace}
\newcommand{\Relabeln}{ReLbl\xspace}
\newcommand{\UpdateEdgeTgtn}{UpdTgt\xspace}
\newcommand{\UpdateEdgeSrcn}{UpdSrc\xspace}

\newcommand{\AddLabelsf}{\ensuremath{\mathsf{\AddLabeln}}}
\newcommand{\RmvLabelsf}     {\ensuremath{\mathsf{\RmvLabeln}}}
\newcommand{\AddObjectsf}    {\ensuremath{\mathsf{\AddObjectn}}}
\newcommand{\RmvObjectsf}    {\ensuremath{\mathsf{\RmvObjectn}}}
\newcommand{\AddEdgesf}      {\ensuremath{\mathsf{\AddEdgen}}}
\newcommand{\AddPropsf}      {\ensuremath{\mathsf{\AddPropn}}}
\newcommand{\RmvPropsf}      {\ensuremath{\mathsf{\RmvPropn}}}
\newcommand{\Relabelsf}      {\ensuremath{\mathsf{\Relabeln}}}
\newcommand{\UpdateEdgeTgtsf}{\ensuremath{\mathsf{\UpdateEdgeTgtn}}}
\newcommand{\UpdateEdgeSrcsf}{\ensuremath{\mathsf{\UpdateEdgeSrcn}}}
\newcommand{\nextrelsf}{\ensuremath{\mathsf{\nextreln}}\xspace}
\newcommand{\startrelsf}{\ensuremath{\mathsf{\startreln}}\xspace}
\newcommand{\mnextrelsf}{\ensuremath{\mathsf{\mnextreln}}\xspace}
\newcommand{\mstartrelsf}{\ensuremath{\mathsf{\mstartreln}}\xspace}

\newcommand{\nextrel}{\nextrelsf}
\newcommand{\startrel}{\startrelsf}
\newcommand{\mnextrel}{\mnextrelsf}
\newcommand{\mstartrel}{\mstartrelsf}

\newcommand{\RmvObject}{\RmvObjectsf}

\newcommand{\Relabel}{\Relabelsf}
\newcommand{\UpdateEdgeTgt}{\UpdateEdgeTgtsf}

\DeclareMathOperator*{\argmax}{arg\,max}

\usepackage{nth}

\usepackage[most]{tcolorbox}
\tcbuselibrary{breakable}
\newcounter{takeaway}

\newenvironment{keytakeaway}[1][1]{
\refstepcounter{takeaway}
\textbf{Takeaway \#\thetakeaway:}
}{}

\newcounter{rqnumber}
\newcommand{\rquest}{\refstepcounter{rqnumber}\textbf{RQ\therqnumber.}\xspace}
\newcounter{rqnumberappendix}

\newcommand{\rquestapp}{\refstepcounter{rqnumberappendix}\textbf{RQ\Alph{rqnumberappendix}.}\xspace}

\newcommand{\meta}[1]{\ensuremath{#1^*}\xspace}

\usepackage{minted}

\usepackage{tikz}
\usepackage{colortbl}

\usepackage{multirow} 
\usepackage[normalem]{ulem} 
\usepackage{xcolor}
\usepackage[most]{tcolorbox}
\tcbuselibrary{minted}

\usetikzlibrary{positioning,calc,arrows.meta,fit,
                shapes.geometric,shapes.misc}

\newsavebox{\selvestebox}

\newcommand{\propdata}[3]{%
  \begin{tcolorbox}[
    colback=white,
    colframe=black!20!white,
    colbacktitle=black!5!white,
    coltitle=black,
    center title,
    size=title,
    left=1mm,
    right=1mm,
    width=#1,
    title=#2
  ]
  #3
\end{tcolorbox}
}

\newcommand{\propdataclear}[3]{%
  \begin{tcolorbox}[
    colback=white,
    colframe=black!5!white,
    colbacktitle=black!5!white,
    coltitle=black!50!white,
    center title,
    size=title,
    left=1mm,
    right=1mm,
    width=#1,
    title=#2
  ]
  #3
\end{tcolorbox}
}

\newtcblisting{propdatacode}[1]{%
      listing only,
      listing engine=minted,
      minted language=prolog,
      colback=white,
      colframe=black!20!white,
      colbacktitle=black!5!white,
      coltitle=black,
      center title,
      size=title,
      left=1mm,
      right=1mm,
      width={#1}}

\newcommand{\propdataempty}[2]{%
  \begin{tcolorbox}[
    colframe=black!20!white,
    colback=black!5!white,
    coltitle=black,
    center title,
    size=title,
    left=1mm,
    right=1mm,
    width=#1,
  ]
  #2
\end{tcolorbox}
}

\newcommand{\nodename}[2]{\textit{#1}: \textbf{#2}}
\newcommand{\nodeprop}[2]{\textbf{#1}: \textsc{#2}}
\newcommand{\propedge}[8]{
  \draw[-latex, thick] (#1.#3) ..controls ($(#1.#3)!.25!(#2.#3)+(#4,0)$) and ($(#1.#3)!.75!(#2.#7)+(#4,0)$) .. node[#5,sloped,rotate=#6] {
    #8
  } (#2.#7);
}
\newcommand{\propedgered}[8]{
  \draw[-latex, thick,red] (#1.#3) ..controls ($(#1.#3)!.25!(#2.#3)+(#4,0)$) and ($(#1.#3)!.75!(#2.#7)+(#4,0)$) .. node[#5,sloped,rotate=#6] {
    #8
  } (#2.#7);
}

\newcommand{\propedgeleftup}[3]{\propedge{#1}{#2}{west}{-0.15cm}{above}{0}{west}{#3}}
\newcommand{\propedgerightup}[3]{\propedge{#1}{#2}{east}{+0.15cm}{below}{0}{east}{#3}}
\newcommand{\propedgerightupred}[3]{\propedgered{#1}{#2}{east}{+0.15cm}{below}{0}{east}{#3}}

\graphicspath{{figures/}}
\usepackage{booktabs}

\begin{document}
\title{Transformations for Evolving Property Graph Schemas}

\author{Gauvain Devillez}
\affiliation{%
  \institution{University of Mons}
  \city{Mons}
  \country{Belgium}
}
\email{gauvain.devillez@umons.ac.be}

\author{Stefania Dumbrava}
\affiliation{%
  \institution{ENSIIE, INRIA, IRIF, SAMOVAR}
  \city{Evry}
  \country{France}
}
\email{stefania.dumbrava@ensiie.fr}

\author{Angela Bonifati}
\affiliation{%
  \institution{Lyon 1 University, CNRS Liris, IUF}
  \city{Lyon}
  \country{France}
}
\email{angela.bonifati@univ-lyon1.fr}

\begin{abstract}
Property graph databases are widely used to represent complex and evolving data, yet
 systematic support for property graph schema evolution remains limited.  In
practice, schema transformations are typically defined manually, 
coupled to specific application contexts, and difficult to reuse across
schemas or evolution scenarios.

We present GRAFT, a logic-based framework that models property graph schema evolution as reusable, order-constrained meta-transformations derived from atomic edits. Schema evolution is formulated as the exploration of a finite meta-graph with schemas as nodes and grounded meta-transformations as edges. To ensure tractability, GRAFT combines similarity-guided search and pruning, guaranteeing duplication-freeness, termination, and correctness.

An experimental evaluation on four benchmark and real-world property graph schema evolution scenarios shows that GRAFT efficiently computes high-quality schema transformation sequences. Using greedy exploration, GRAFT reaches the exact target schema on most datasets, producing stable transformation sequences while keeping runtimes low. A qualitative study on both real-world and synthetic large-scale datasets further shows the quality and robustness of the obtained reusable meta-transformations.
\end{abstract}

\maketitle

\pagestyle{\vldbpagestyle}
\begingroup\small\noindent\raggedright\textbf{PVLDB Reference Format:}\\
\vldbauthors. \vldbtitle. PVLDB, \vldbvolume(\vldbissue): \vldbpages, \vldbyear.\\
\href{https://doi.org/\vldbdoi}{doi:\vldbdoi}
\endgroup
\begingroup
\renewcommand\thefootnote{}\footnote{\noindent
This work is licensed under the Creative Commons BY-NC-ND 4.0 International License. Visit \url{https://creativecommons.org/licenses/by-nc-nd/4.0/} to view a copy of this license. For any use beyond those covered by this license, obtain permission by emailing \href{mailto:info@vldb.org}{info@vldb.org}. Copyright is held by the owner/author(s). Publication rights licensed to the VLDB Endowment. \\
\raggedright Proceedings of the VLDB Endowment, Vol. \vldbvolume, No. \vldbissue\ %
ISSN 2150-8097. \\
\href{https://doi.org/\vldbdoi}{doi:\vldbdoi} \\
}\addtocounter{footnote}{-1}\endgroup

\ifdefempty{\vldbavailabilityurl}{}{
\vspace{.3cm}
\begingroup\small\noindent\raggedright\textbf{PVLDB Artifact Availability:}\\
The source code, data, and/or other artifacts have been made available at \url{https://github.com/GauvainD/GRAFT}.
\endgroup
}

\section{Introduction}\label{sec:intro}

Property graph databases have emerged as a core technology for representing and
querying richly structured and highly interconnected data. Their flexible schema
model, supporting heterogeneous node and edge types and multi-valued properties, has led to widespread adoption in domains such as social
networks, bioinformatics, intelligent transportation, and fraud
detection. Unlike the rigid and highly normalized schemas of
relational databases, property graph schemas often evolve incrementally, without
centralized control, and are frequently inferred \emph{post hoc} from data. This
is common in data integration, where graph-shaped data is
ingested from heterogeneous sources with diverse semantics and structural
granularity.

For relational data, schema evolution has benefited from a mature body
of work on versioning, migration, and mapping management. In particular, the
model management paradigm introduced by Bernstein~\cite{DBLP:conf/cidr/Bernstein03}
treats mappings and meta-mappings as first-class objects, enabling systematic
reuse and composition. Comparable principles, however, have not been established
for property graph schemas. Graph databases are commonly deployed in environments
where multiple schema versions coexist, downstream applications impose
compatibility requirements, and schemas are continuously refined as data is
enriched by analytics, machine learning pipelines, or business
logic~\cite{DBLP:journals/cacm/SakrBVIAAAABBDV21}. Schema evolution in
these settings relies on ad hoc, manually specified transformations 
coupled with specific applications, which are difficult to reuse and costly to
maintain.

\begin{figure}[t!]
\centering
\footnotesize
\resizebox{0.95\linewidth}{!}{%
  \begin{tikzpicture}
  \matrix [anchor=north west] {
    \node (source) {
        \begin{tikzpicture}
  \node (location) {
    \propdata{2.5cm}{\nodename{}{Place}}{
      \nodeprop{id}{ID}\\
      \nodeprop{name}{Long String}
    }
  };
  \node[below=-.1cm of location] (person) {
    \propdata{2.5cm}{\nodename{}{User}}{
      \nodeprop{id}{ID}\\
      \nodeprop{firstName}{String}\\
      \nodeprop{lastName}{String}
    }
  };
  \node[below=-.1cm of person] (modo) {
    \propdata{2.5cm}{\nodename{}{Moderator}}{
      \nodeprop{id}{ID}\\
      \nodeprop{firstName}{String}\\
      \nodeprop{lastName}{String}
    }
  };
  \node[below=-.1cm of modo] (forum) {
    \propdata{2.5cm}{\nodename{}{Forum}}{
      \nodeprop{id}{ID}\\
      \nodeprop{title}{Long String}
    }
  };
  \propedgeleftup{forum}{person} {
    \propdata{2cm}{\nodename{}{hasMember}}{
      \nodeprop{creation}{Date}
    }
  }
  \propedgerightup{forum}{modo} {
    \propdataempty{2cm}{\nodename{}{hasModerator}}{
    }
  }
  \propedgerightup{person}{location} {
    \propdataempty{2cm}{\nodename{}{isLocatedIn}}{
    }
  }
\end{tikzpicture}
    };&
    \node (schema2) {
        \begin{tikzpicture}
  \node (location) {
    \propdata{2.5cm}{\nodename{}{\textcolor{red}{Location}}}{
      \nodeprop{id}{ID}\\
      \nodeprop{name}{Long String}
    }
  };
  \node[below=-.1cm of location] (person) {
    \propdata{2.5cm}{\nodename{}{\textcolor{red}{Person}}}{
      \nodeprop{id}{ID}\\
      \nodeprop{firstName}{String}\\
      \nodeprop{lastName}{String}
    }
  };
  \node[below=-.1cm of person] (modo) {
    \propdataclear{2.5cm}{\nodename{}{Moderator}}{
      \textcolor{black!50!white}{
        \nodeprop{id}{ID}\\
        \nodeprop{firstName}{String}\\
        \nodeprop{lastName}{String}
      }
    }
  };
  \node[below=-.1cm of modo] (forum) {
    \propdata{2.5cm}{\nodename{}{Forum}}{
      \nodeprop{id}{ID}\\
      \nodeprop{title}{Long String}
    }
  };
  \propedgeleftup{forum}{person} {
    \propdata{2cm}{\nodename{}{hasMember}}{
      \nodeprop{creation}{Date}
    }
  }
  \propedgerightupred{forum}{person} {
    \propdataempty{1.9cm}{\nodename{}{{\color{red}hasModerator}}}{
    }
  }
  \propedgerightup{person}{location} {
    \propdataempty{1.9cm}{\nodename{}{isLocatedIn}}{
    }
  }
\end{tikzpicture}
    };\\
  };
  \draw[-latex, thick, blue]
    (source.north)
    .. controls ($(source.north)!.25!(schema2.north)+(0,.25)$)
    and ($(source.north)!.75!(schema2.north)+(0,.25)$)
    .. node[above, sloped] {
      \propdata{11cm}{}{
      \centering
      \footnotesize
        {\color{blue}\Relabelsf}(\textbf{User},\,{\color{red}\textbf{Person}}) $\circ$
        {\color{blue}\Relabelsf}(\textbf{Place},\,{\color{red}\textbf{Location}}) $\circ$
        {\color{blue}\RmvObjectsf}(\textbf{Moderator}) $\circ$
        {\color{blue}\UpdateEdgeTgtsf}({\color{red}\textbf{hasModerator}},\,\textbf{User})
      }
    } (schema2.north);
\end{tikzpicture}

}
\caption{Property graph schema evolution. Red annotations highlight changes between two
versions.}\label{fig:meta-graph}
\end{figure}
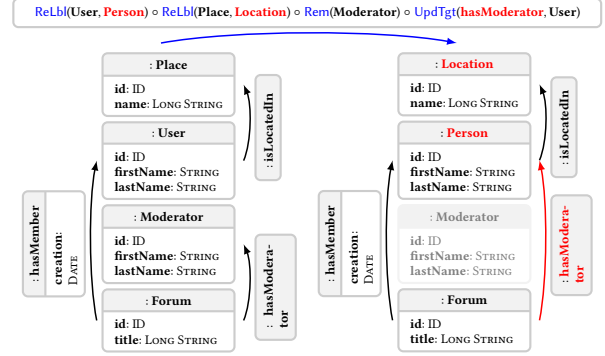

A key reason for this gap is the nature of property graph schemas
themselves. Rather than being flat collections of relations or classes, property
graph schemas are \emph{typed graphs}. Structural and typing constraints are
therefore tightly intertwined: modifying a node or edge type may invalidate
adjacent schema elements, while changes in typing often require coordinated
structural updates. Hence, seemingly simple schema modifications, e.g., 
renaming a node type, merging types, or removing an edge, cannot be treated in
isolation without risking inconsistencies. This contrasts with the
relational model, where tables and attributes can often evolve independently. In addition, property graph schemas 
can be descriptive instead of simply prescriptive as well as partial, leading to part of the instances adhering to a schema and evolving as needed ~\cite{DBLP:journals/pacmmod/AnglesBD0GHLLMM23}. This makes property graph schema evolution even more compelling. 

Furthermore, the choice of the schema can have significant impact on query performance as shown in~\cite{DBLP:conf/icde/AlotaibiLQEO21,DBLP:conf/dasfaa/XuLZHHLWZ24,DBLP:journals/pacmmod/SharmaGGL25}.

Schema transformations thus play a central role in managing evolving property
graphs. Although these transformations are still defined manually, experience
across application domains shows that they exhibit recurring structural patterns.
This raises a fundamental question: \emph{can property graph schema
evolution be automated by identifying, abstracting, and reusing existing
transformation patterns?} Addressing this question would enable systematic schema
evolution, reduce manual effort, and support scalable graph data integration.

\paragraph{Motivating Example.}
Consider an online platform whose schema resembles the LDBC Social Network
Benchmark~\cite{DBLP:journals/corr/abs-2001-02299}. As the schema evolves, user
roles and location entities are progressively standardized across services.
Figure~\ref{fig:meta-graph} illustrates two successive schema versions, with red
annotations indicating the changes from version~1 to version~2, which involve renaming nodes and changing edge endpoints. 
These changes implicitly require a
coordinated sequence of schema edits, whose order and interaction matter. This
example highlights the need for principled mechanisms to represent, compose, and
reuse property graph schema transformations.

\paragraph{Problem Statement.}
We study the problem of evolving a property graph schema using a fixed set of transformations. For a source schema $S$ and a target schema $\mathcal{T}$, we compute a sequence of transformations of $S$ that produces $\mathcal{T}$ if possible, and otherwise a schema $\mathcal{T}'$ that maximizes similarity to $\mathcal{T}$ for a given similarity measure.

The schema transformations considered in this work correspond to concrete schema
evolution steps observed between successive schema versions or across related
schemas. Such transformations naturally arise in data lakes, where heterogeneous
graph data is ingested without a priori schema and evolves as new data arrives,
and in large graph ecosystems~\cite{DBLP:journals/cacm/SakrBVIAAAABBDV21}, where
schema variants proliferate as data is enriched by downstream processes. To
enable reuse beyond their original context, we abstract concrete schema
transformations into \emph{meta-transformations}, which generalize schema changes
via variable substitution and explicit ordering constraints. This abstraction
supports the systematic reuse and composition of transformation patterns across
schemas and evolution scenarios.

Addressing this problem requires (i) abstracting concrete schema transformations
into reusable patterns, (ii) composing such patterns to generate new
transformation sequences, and (iii) efficiently exploring the resulting space of
candidate schemas. The central challenge lies in navigating this large search
space while ensuring duplication-freeness, termination, correctness, and efficiency.

We present \graft, a logic-based framework for reusable property graph schema evolution, which we position in the context of related work (§2), and the recent property graph model and schema formalism (§3). Our main contributions are:

\begin{itemize}
    \item \emph{A formal model of reusable schema evolution (§4.1–§4.2).}
    We formalize property graph schema evolution via atomic edits and meta-transformations that abstract evolution steps into reusable, order-constrained transformation patterns. 

    \item \emph{Duplication-free and terminating transformation derivation (§4.2).}
    We present algorithms to ground and enumerate transformations, with guarantees of duplication-freeness, termination, finiteness, and worst-case complexity bounds.

    \item \emph{A meta-graph exploration algorithm (§4.3–§4.4).}
    We model evolution as search in a finite meta-graph of schemas and propose \graft, a pruning-aware exploration algorithm with guarantees of soundness and completeness (without pruning), and minimality under breadth-first search.

    \item \emph{Complexity analysis and search space reduction (§4.3–§4.4).}
    We analyze the exponential worst-case search space and introduce clique-based compression to reduce permutation exploration, improving scalability.

    \item \emph{Comprehensive empirical validation (§5).}
        We assess \graft on four pre-generated schema evolution scenarios from the iBench benchmark~\cite{DBLP:journals/pvldb/ArocenaGCM15}, a real-world scenario based on the ICIJ Paradise Papers journalistic investigation dataset~\cite{icijdb}, and a large schema we generate with iBench to stress-test scalability. We evaluate its ability to recover target schemas, the quality and minimality of the resulting transformation sequences, and its scalability under varying branching factors and pruning strategies. Our qualitative study highlights the stability and reusability of derived meta-transformations on a realistic use case. A further experiment evaluates the runtime and data-level cost of applying the resulting transformation sequences.
\end{itemize}
We then conclude and outline future research perspectives (§6).

\section{Related Work}\label{sec:related}

Schema evolution and transformation have been extensively studied in relational databases and ontologies, with a focus on schema mappings, transformation automation, and reuse. In the relational setting, schema mappings such as Global-as-View (GAV) and Local-as-View (LAV) initially required manual specification, which motivated research on automatic mapping generation~\cite{DBLP:conf/vldb/HaasKWY97,DBLP:conf/pods/DuschkaG97} as well as reuse and generalization techniques. Atzeni et al.~\cite{DBLP:journals/pvldb/AtzeniBPT19} introduced  \emph{meta-mappings} that enable the systematic reuse of mappings across different contexts and were later developed in the GAIA framework. However, these approaches are tailored to relational data.

Bernstein et al. ~\cite{DBLP:conf/cidr/Bernstein03,phil-sigmod07} introduced model management systems to create, compile and reuse mappings across schemas represented in a wide range of meta-models. This pioneering work defined transformations as functional mapping constraints and defined a set of meta-model operators, such as Match, ModelGen, Merge and Compose, the last of which supports mapping reuse. Our work builds upon these lines, while focusing on specific duplicate-free meta-transformations and schema evolution for property graphs with termination and correctness guarantees.   

Schema-modification operators (SMOs), such as those in the PRISM/PRISM++
workbench~\cite{DBLP:journals/pvldb/CurinoMZ08,DBLP:journals/pvldb/CurinoMDZ10},
evolve a relational schema through a fixed catalogue of atomic operators with
automatic query and update rewriting and data migration. Such SMOs, like the meta-mappings and model-management operators
above, act over tables, functional dependencies, or generic meta-models. In contrast, \graft abstracts
transformations directly over typed property-graph schemas, where structure and typing are coupled. A relational operator can drop a column without affecting other tables,
whereas a property-graph edge endpoint is itself typed by a node type, so renaming or
removing a node type automatically affects the typing of incident edges. This forces an ordering
on edits (e.g., redirecting \texttt{hasModerator} must precede removing
\texttt{Moderator}), which we account for explicitly with dedicated \mstartrel/\mnextrel relations, as explained in Section~\ref{ssec:meta-transformations}. A side-by-side relational encoding of the running example and an operator-level comparison with relational SMOs are given in~\cite[App. A]{GRAFTAppendix}.

Automated schema transformation was also explored through instance-driven approaches~\cite{DBLP:journals/tods/BonifatiCCT19,DBLP:conf/ssdbm/MaziluPFK19}, and through learning-based techniques~\cite{DBLP:conf/pods/CateK0T18,DBLP:conf/birthday/AlexeT13}, which infer mappings from data instances or query answers. While powerful, these methods primarily target the discovery of mappings rather than the optimization of schema evolution sequences. Similarly, schema and ontology matchers such as
COMA++~\cite{DBLP:conf/sigmod/AumuellerDMR05} compute element-level
correspondences between a source and a target schema, but not the ordered,
executable edit sequences that realize the evolution. \graft is complementary to
such matchers, whose candidate correspondences it could  compose and validate into reusable transformations.

\emph{Property graph schemas} are a relatively recent notion, with early efforts focusing on schema definition, typing, and validation. Bonifati et al.~\cite{DBLP:conf/er/BonifatiFGHOV19} introduce an extension of openCypher for expressing property graph schemas, as well as a framework based on graph rewriting for maintaining consistency under updates. Related works~\cite{DBLP:journals/pvldb/BonifatiMR24,DBLP:journals/pvldb/BonifatiRMFE24,DBLP:journals/pvldb/Bonifati25} leverage a declarative approach for expressing \emph{schema-less graph transformations} based on the Graph Pattern Calculus (GPC)~\cite{DBLP:conf/pods/FrancisGGLMMMPR23}, common to the GQL~\cite{GQL-ISO} and SQL/PGQ~\cite{SQL-PGQ} standard graph query languages. Mechanisms for expressing and implementing \emph{graph views}~\cite{DBLP:journals/pacmmod/HanI24,DBLP:journals/sigmod/HanI25} have also been recently introduced to encode derived graph structures. These approaches emphasize flexibility and support query-based schema evolution; however, they do not tackle the systematic reuse or composition of graph schema transformations over time. In parallel, community-driven efforts toward \emph{property graph schema standards} (e.g., PG-Schema~\cite{DBLP:journals/pacmmod/AnglesBD0GHLLMM23} and PG-Keys~\cite{DBLP:conf/sigmod/AnglesBDFHHLLLM21}) highlight the increasing importance of formal schema management and interoperability for property graphs.

To the best of our knowledge, \textsc{GRAFT} is the first system to address \emph{property graph schema evolution through meta-transformations}. By constructing a meta-graph model that captures high-level transformation patterns, GRAFT efficiently computes minimal evolution paths between schemas. Unlike prior probabilistic or heuristic-based works~\cite{DBLP:conf/icde/KimmigMMG17,DBLP:journals/tkde/KimmigMMG19}, \textsc{GRAFT} reuses meta-transformations and provides automatic transformation generation with guarantees of termination, correctness, losslessness, and minimality.

\section{Preliminaries}\label{sec:prelim}

We first define property graphs~\cite{DBLP:conf/amw/Angles18,DBLP:series/synthesis/2018Bonifati} and the considered schema fragment. Let $\mathcal{L}$, $\mathcal{K}$,
$\mathcal{D}$, and $\mathcal{T}$ be pairwise disjoint sets of labels, keys,
values, and data types. $\powset{X}$ is the power set of a set $X$.

\begin{definition}[Property Graph]
A property graph $\propgraph{}$ is a tuple
$(\mathcal{V}, \mathcal{E}, \rho, \lambda, \sigma)$ where:
\begin{itemize}
    \item $\mathcal{V}$ and $\mathcal{E}$ are finite, disjoint sets of vertices and edges;
    \item $\rho : \mathcal{E} \rightarrow \mathcal{V} \times \mathcal{V}$ assigns each edge its source and target;
    \item $\lambda : \mathcal{V} \cup \mathcal{E} \rightarrow \powset{\mathcal{L}}$ assigns a finite set of labels;
    \item $\sigma : \mathcal{V} \cup \mathcal{E} \rightarrow \powset{\mathcal{K} \times \mathcal{D}}$
    assigns a finite set of properties.
\end{itemize}
\end{definition}

We represent graph databases as property graphs.
Figure~\ref{fig:propg} shows an example inspired by the LDBC SNB
schema~\cite{DBLP:journals/corr/abs-2001-02299}.

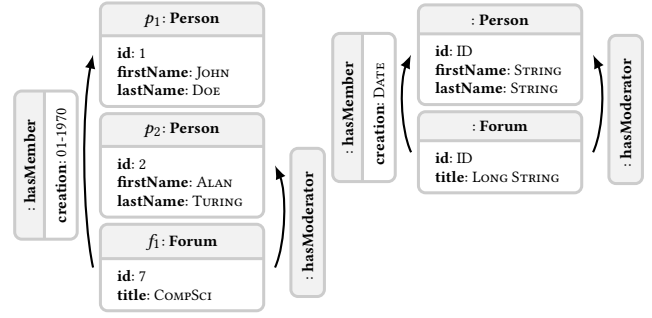
\begin{figure}[!htbp]
    \centering
    {\scriptsize
      \begin{tikzpicture}
  \node (person1) {
    \propdata{2.2cm}{\nodename{$p_1$}{Person}}{
      \nodeprop{id}{1}\\
      \nodeprop{firstName}{John}\\
      \nodeprop{lastName}{Doe}
    }
  };
  \node[below=-.1cm of person1] (person2) {
    \propdata{2.2cm}{\nodename{$p_2$}{Person}}{
      \nodeprop{id}{2}\\
      \nodeprop{firstName}{Alan}\\
      \nodeprop{lastName}{Turing}
    }
  };
  \node[below=-.1cm of person2] (forum) {
    \propdata{2.2cm}{\nodename{$f_1$}{Forum}}{
      \nodeprop{id}{7}\\
      \nodeprop{title}{CompSci}
    }
  };
  \propedgeleftup{forum}{person1} {
    \propdata{2cm}{\nodename{}{hasMember}}{
      \nodeprop{creation}{01-1970}
    }
  }
  \propedgerightup{forum}{person2} {
    \propdataempty{2cm}{\nodename{}{hasModerator}}{
    }
  }

  \node[right=1.8 of person1] (persont) {
    \propdata{2.2cm}{\nodename{}{Person}}{
      \nodeprop{id}{ID}\\
      \nodeprop{firstName}{String}\\
      \nodeprop{lastName}{String}
    }
  };
  \node[below=-.1cm of persont] (forumt) {
    \propdata{2.2cm}{\nodename{}{Forum}}{
      \nodeprop{id}{ID}\\
      \nodeprop{title}{Long String}
    }
  };
  \propedgeleftup{forumt}{persont} {
    \propdata{2cm}{\nodename{}{hasMember}}{
      \nodeprop{creation}{Date}
    }
  }
  \propedgerightup{forumt}{persont} {
    \propdataempty{2cm}{\nodename{}{hasModerator}}{
    }
  }
\end{tikzpicture}

    }
    \caption{Example property graph (left) and its schema (right).}\label{fig:propg}
\end{figure}

We consider a finite fragment of PG-Schema~\cite{DBLP:journals/pacmmod/AnglesBD0GHLLMM23}
where schemas consist only of node and edge types, and connectivity
is defined by typing the endpoints of edge types. 
Additionally supporting schema constraints would require extra validity predicates and dedicated edit operations, left to future work. Typing is functional: every node and edge conforms to exactly one schema type.

\begin{definition}[Base Type]
A base type is a pair $(L, R)$ where
$L \subseteq \mathcal{L}$ is a finite set of labels and
$R \subseteq \mathcal{K} \times \mathcal{T}$ is a finite record.
\end{definition}

Let $o$ be a graph object (node/edge) with label set $\lambda(o)$ and properties
$\sigma(o)$. We say that $o$ conforms to a base type $(L, R)$ if
$\lambda(o) = L$ and,
for every $(k, \tau) \in R$, there exists $(k, v) \in \sigma(o)$
with value $v$ of type $\tau$. For example, in Figure~\ref{fig:propg}, the nodes \textsc{Alan Turing} and \textsc{John Doe}
conform to the base type:
{\small
\[\verb|({Person}, {id:ID, firstName:STRING, lastName:STRING})|.\]}

Beyond element-level typing, a core property graph schema specifies how base types
may be connected.

\begin{definition}[Core Property Graph Schema]
A core property graph schema is a tuple
$(\mathcal{V}_S, \mathcal{E}_S, \rho_S)$ where:
\begin{itemize}
    \item $\mathcal{V}_S$ is a finite set of node base types;
    \item $\mathcal{E}_S$ is a finite set of edge base types;
    \item $\rho_S : \mathcal{E}_S \rightarrow \mathcal{V}_S \times \mathcal{V}_S$
    specifies, for each edge base type, its source and target
    node base types.
\end{itemize}
\end{definition}

\paragraph{Schema Conformance.}
Let $\propgraph{}$
be a property graph and
$S$ be its
core property graph schema.
We say that $\propgraph{}$ conforms to $S$ if every node and
edge of $\propgraph{}$ conforms to exactly one base type in
$\mathcal{V}_S$ and $\mathcal{E}_S$, respectively, and for every edge
$e$ of base type $t_e$ with $\rho(e) = (v_1, v_2)$ and
$\rho_S(t_e) = (t_1, t_2)$, the nodes $v_1$ and $v_2$ conform to
$t_1$ and $t_2$.

The graph on the left of Figure~\ref{fig:propg} conforms to the
schema on the right. However, the same data can be represented by different
schemas; for example, the left schema in
Figure~\ref{fig:meta-graph} encodes the data in
Figure~\ref{fig:propg} using a different structure. Next, we study transformations between such schemas to support their principled evolution. 
Henceforth, selected proofs are sketched due to space constraints.
Complete proofs are given in an appendix~\cite[App. K]{GRAFTAppendix}.


\section{Property Graph Schema Evolution}\label{sec:prop-graph-schema}

We present a property graph schema evolution framework based on reusable meta-transformations (Sections~\ref{ssec:meta-operations}-\ref{ssec:meta-transformations}) and similarity-guided exploration of a meta-graph search space (Sections~\ref{ssec:exploration}-\ref{ssec:graft}).


\subsection{Edit Operations and Meta-Edit Operations}\label{ssec:meta-operations}


Evolving across heterogeneous schemas requires formal restructuring mechanisms, typically captured as schema transformations.

For example, transforming the source schema into the target schema in
Figure~\ref{fig:meta-graph} requires changing the label \texttt{Place}
to \texttt{Location} for a single node. We refer to such
elementary changes as edit operations, which we define below.

\begin{definition}[Valid edit operations and transformations]
An \emph{edit operation} is an instance of a schema
operation in Table~\ref{fig:edit-ops}. A \emph{transformation}
is a finite ordered sequence of edit operations.

Given a schema $\aschema$, an edit operation is \emph{valid}
if it inserts only absent elements and removes only present ones.
Given a transformation $\atransfo=\{e_n \circ \dots \circ e_0\}$, we say
that it is \emph{valid} with respect to a schema \aschema if all of its edit
operations are valid with respect to \aschema and, for each edit operation
$e_k \in \atransfo$, it is valid with respect to the schema obtained from
\aschema by applying the $k-1$ first edit operations of \atransfo.
\end{definition}

\begin{table}[t]
\caption{Schema edit operations.}\label{fig:edit-ops}
\centering
\scriptsize
\begin{tabular}{@{}ll@{}}
\toprule
\textbf{Operation} & \textbf{Description} \\
\midrule

$\AddLabelsf(o,l)$, 
$\RmvLabelsf(o,l)$ 
& Add or remove a label from an object (vertex or edge). \\[3pt]

$\AddObjectsf(o)$, 
$\RmvObjectsf(o)$ 
& Add or remove an object (vertex or edge). \\[3pt]

$\AddEdgesf(e,src,tgt)$ 
& Add an edge type between two vertex types. \\[3pt]

$\AddPropsf(o,p,t)$, 
$\RmvPropsf(o,p,t)$ 
& Add or remove a typed property of an object. \\[3pt]

$\Relabelsf(l,l')$ 
& Replace an existing label by another. \\[3pt]

$\UpdateEdgeTgtsf(e,tgt)$,
$\UpdateEdgeSrcsf(e,src)$ 
& Update the target/source of an edge.  \\
\bottomrule
\end{tabular}
\end{table}

\begin{example}[Valid Schema Transformations]\label{ex:valid-transfo}
Consider the transformation illustrated in Figure~\ref{fig:meta-graph}.
It consists of four edit operations: two update the target of an edge
prior to removing its previous target, and two relabelings. For the left schema in Figure~\ref{fig:meta-graph},
$\RmvObjectsf(Website)$ is invalid, since the element to be removed
is not present. Moreover, even if each edit operation is 
valid, their composition may not be: for instance,
$\RmvObjectsf(Moderator) \circ \Relabelsf(Moderator, User)$ is invalid,
as it would remove the node \texttt{Moderator}
after its label has been changed.
\end{example}

As described in the introduction, the problem is to reuse schema
transformations. However, these require an exact match between labels,
properties, node types, etc.

To enable reuse across different schemas, we abstract concrete schemas and transformations by systematically replacing constants (e.g., labels, node types, and property names) with variables. This generalization produces meta-schemas and meta-transformations that capture structural transformation patterns independently of specific symbol names. These abstractions can then be instantiated over different schemas through suitable variable assignments.

Given two schemas from different property graphs, one can notice similarities. For
example, one might structure a graph database for a school in the same way as
what is shown in our example on social networks. A user corresponds to a teacher and/or a student, a
forum becomes a course, while location would be a classroom.

Edit operations are abstracted by replacing constants with
variables, producing meta-edit operations. Given a schema $\aschema$, grounding
a meta-edit operation assigns its variables to constants from $\aschema$,
producing concrete edit operations. Such assignments must satisfy validity
constraints; for example, a property can be removed only from a node that
contains it. These constraints may depend on the operation itself and on the
transformation in which it appears.

\begin{definition}[Meta-schema and meta-edit operations]
Let $\aschema$ be a schema. A \emph{meta-schema} $\meta{\aschema}$
is obtained by systematically replacing each distinct constant
in $\aschema$ with a fresh variable.
A \emph{meta-(atomic) edit operation} is obtained analogously
by replacing each distinct constant in an atomic edit operation
with a fresh variable.
\end{definition}

Definition~5 abstracts schemas and atomic edits. In the following, we 
do the same for transformation sequences.


\subsection{Meta-Transformations and Grounding}\label{ssec:meta-transformations}

Next, we explain how transformations are abstracted into meta-transformations (Section~\ref{ssec:transfo-abstraction}), how these are grounded into transformations (Section~\ref{ssec:derivation}), and how they are applied to the input schema to produce intermediate and target schemas (Section~\ref{ssec:validation}).

\subsubsection{Transformations and Meta-transformations}\label{ssec:transfo-abstraction}

Given a meta-schema \meta{\aschema} and an assignment of fresh variables to types, one can obtain a concrete schema \aschema. Different variable assignments may produce different concrete schemas from the same meta-schema. The notions of meta-schema and meta-edit operations naturally extend to transformations, giving rise to meta-transformations.

\begin{definition}[Grounding]
Given a meta-schema, meta-edit operation, or meta-transformation, \emph{grounding} is the instantiation of its variables with elements (e.g., types/labels/properties) from the source or target schema to obtain a concrete schema, edit, or transformation.
\end{definition}

Generating meta-transformations from transformations is more involved than for meta-edit operations, as ordering constraints between edit operations must be preserved by grounding. 

Consider the transformation between the source schema and Schema~1 in Figure~\ref{fig:meta-graph}. In this case, the order of edit operations is critical: updating the target of an edge after deleting the \texttt{Moderator} node, or after renaming the \texttt{User} node, would introduce references to symbols that no longer exist.

Accordingly, a meta-transformation must encode the relative ordering of its constituent edit operations. Moreover, not every grounding yields a valid transformation. For example, two renaming operations may be instantiated on the same node: assigning the same label renders the second operation redundant, whereas assigning different labels may invalidate subsequent operations. To ensure that grounding produces valid, well-ordered transformations, each meta-transformation is equipped with two rule sets, $\mstartrel$ and $\mnextrel$. $\mstartrel$ specifies which operations may start a sequence, and $\mnextrel$ defines valid and admissible successor relationships.

These can be evaluated with a Datalog-based engine. To this end, \graft{} is
built on top of Soufflé~\cite{DBLP:conf/cav/JordanSS16} and could also support
integrity constraints expressible in Soufflé's Datalog dialect. Each such
constraint could be expressed as a validity predicate over edits whose
violating operations are rejected at grounding. We leave constraint-preserving
schema evolution to future work.

The exact matching of Section~\ref{ssec:meta-operations} governs how
symbols are compared but does not freeze labels: \Relabel{} (Table~1) rewrites
any label into any other, related or not (e.g., $\textsf{Place} \to
\textsf{Location}$), so relabeling is fully supported. Value domains, in turn,
are not schema-level. A property is given by its name and type, and a unit
conversion such as km~$\to$~miles preserves both, surfacing only at the data
level. Such conversions integrate directly as value-rewriting edit operations
grounded by domain-specific conversion functions, which compose with the
ordering and reuse mechanism formalized below.

\begin{definition}[Meta-transformation]\label{def:meta-transformation}
A \emph{meta-transformation} is a triple
$(\meditset, \mstartrel, \mnextrel)$, where $\meditset$ is a finite set of
meta-edit operations, and $\mstartrel$ and $\mnextrel$ constrain the admissible
start and admissible successor relations between the derived edit operations.
\end{definition}

\begin{figure}[!htbp]
  \centering
  \begin{minipage}[b]{1.0\linewidth}
{\scriptsize
\[
\begin{array}{l}

\textcolor{green!60!black}{\UpdateEdgeTgtsf}(e_{hM},v_U) \leftarrow
\textcolor{black}{\mathsf{SrcEdge}}(e_{hM},\_,x),
\mathsf{TgtEdge}(e_{hM},\_,v_U), x \neq v_U. \\

\textcolor{green!60!black}{\RmvObjectsf}(v_M) \leftarrow
\mathsf{SrcNode}(v_M), \neg \mathsf{TgtNode}(v_M). \\

\textcolor{green!60!black}{\mstartrelsf}(\UpdateEdgeTgtsf(e_{hM},v_U)) \leftarrow
\neg \mathsf{SrcEdge}(e_{hM},\_,v_U), \mathsf{TgtEdge}(e_{hM},\_,v_U). \\

\textcolor{green!60!black}{\mnextrelsf}(\UpdateEdgeTgtsf(e_{hM}{,}v_U){,}\RmvObjectsf(v_M)) {\leftarrow}
\mathsf{SrcEdge}(e_{hM}{,}\_{,}v_M){,} \mathsf{TgtEdge}(e_{hM}{,}\_{,}v_U). \\

\textcolor{blue}{\Relabelsf}(v,\ell) \leftarrow
\mathsf{SrcLabel}(v), \neg \mathsf{TgtLabel}(v), \mathsf{TgtLabel}(\ell). \\

\textcolor{blue}{\mstartrelsf}(\Relabelsf(v_P,v_L)) \leftarrow
\mathsf{SrcLabel}(v_P), \mathsf{TgtLabel}(v_L). \\

\textcolor{blue}{\mnextrelsf}(\Relabelsf(v_P,v_L),\Relabelsf(v_U,v_{Pe})) \leftarrow
\mathsf{SrcLabel}(v_U), \mathsf{TgtLabel}(v_{Pe}).
\end{array}
\]
}
  \end{minipage}
  \begin{minipage}[b]{1.0\linewidth}
  \begin{tikzpicture}
  \matrix [anchor=north west] {
    \node (source) {
        \begin{tikzpicture}
        \node (location) {
            \propdata{2.5cm}{\nodename{}{$v_{P}$}}{
            }
        };
        \node[below=-0.2 of location] (person) {
            \propdata{2.5cm}{\nodename{}{$v_{U}$}}{
            }
        };
        \node[below=-0.2 of person] (modo) {
            \propdata{2.5cm}{\nodename{}{$v_{M}$}}{
            }
        };
        \node[below=-0.2 of modo] (forum) {
            \propdata{2.5cm}{\nodename{}{Forum}}{
            \nodeprop{id}{ID}\\
            \nodeprop{title}{String}
            }
        };
        \propedgerightup{forum}{modo} {
            \propdataempty{2cm}{\centering\nodename{}{$e_{hM}$}}{
            }
        }
        \end{tikzpicture}
    };&
    \node (schema2) {
        \begin{tikzpicture}
        \node (location) {
            \propdata{2.5cm}{\nodename{}{\textcolor{red}{$v_{L}$}}}{
            }
        };
        \node[below=-0.2 of location] (person) {
            \propdata{2.5cm}{\nodename{}{\textcolor{red}{$v_{Pe}$}}}{
            }
        };
        \node[below=-0.2 of person] (modo) {
            \propdataclear{2.5cm}{\nodename{}{$v_{M}$}}{
            }
        };
        \node[below=-0.2 of modo] (forum) {
            \propdata{2.5cm}{\nodename{}{Forum}}{
            \nodeprop{id}{ID}\\
            \nodeprop{title}{String}
            }
        };
        \propedgerightupred{forum}{person} {
            \propdataempty{2cm}{\centering\nodename{}{{\color{red}$e_{hM}$}}}{
            }
        }
        \end{tikzpicture}
    };\\
  };
  \draw[-latex, thick, blue]
    (source.north)
    .. controls ($(source.north)!.25!(schema2.north)+(0,.25)$)
    and ($(source.north)!.75!(schema2.north)+(0,.25)$)
    .. node[above, sloped] {
      \propdata{8cm}{}{
      \scriptsize
        {{\color{blue}\Relabel}(\textbf{$v_{P}$},\,{\color{red}\textbf{$v_{L}$}}) $\circ$
        {\color{blue}\Relabel}(\textbf{$v_{U}$},\,{\color{red}\textbf{$v_{Pe}$}}) $\circ$
        {\color{green!60!black}\RmvObject}({\color{red}\textbf{$v_{M}$}}) $\circ$
  {\color{green!60!black}\UpdateEdgeTgt}({\color{red}\textbf{$e_{hM}$}},\,\textbf{$v_U$})
      }}
    } (schema2.north);
\end{tikzpicture}


\end{minipage}
\caption{Meta-schema and meta-transformations (blue and green) derived from the source schema in Figure~\ref{fig:meta-graph}.\label{fig:meta-schema}}
\end{figure}
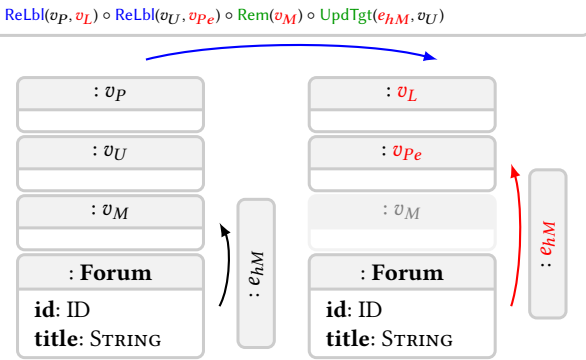

\begin{example}[Meta-Transformation Grounding and Composition]\label{ex:grounding}
  Figure~\ref{fig:meta-schema} illustrates the construction. The meta-transformations shown at the top (green and blue) jointly generate the transformation sequence of Figure~\ref{fig:meta-graph}. The green meta-transformation produces \UpdateEdgeTgt{} and \RmvObject{} operations, requiring each sequence to start with \UpdateEdgeTgt{} and be followed by \RmvObject{} removing the former edge target. The blue meta-transformation generates \Relabel{} operations, starting with one \Relabel{} and optionally continuing with additional \Relabel{} operations. The \Relabel{} operations are unordered, and their number is unconstrained. The meta-edit sequence under the rules is valid and transforms the left meta-schema into the right one, corresponding to the concrete schemas in Figure~\ref{fig:meta-graph}, where constants are replaced by variables and unchanged elements are omitted.

Grounding the meta-schema by assigning values to its variables produces a concrete schema. The same meta-transformations and meta-edit operations can be applied to meta-schemas containing additional elements (e.g., nodes, edges, properties); elements not referenced by the transformation remain unchanged.
\end{example}

\subsubsection{Deriving Meta-Transformations}\label{ssec:derivation}

A meta-transformation is obtained by abstracting a concrete
transformation into meta-edit operations and defining admissible
ordering constraints via the relations $\mstartrelsf$ and $\mnextrelsf$.
These determine the valid grounded edit operations and the
transformation sequences they induce over a source schema.
We describe how such a meta-transformation is grounded,
evaluated, and used to derive admissible sequences, and discuss the
properties of the resulting search space.

The use of the \mnextrel rules allows a meta-transformation to generate
multiple edit operations at once. This raises two potential issues: edit
operations may be duplicated within a transformation, and the evaluation may
yield infinitely many sequences.

We address this by imposing two well-defined restrictions: the domain of each
variable in the meta-transformation is restricted to values from either the
source or the target schema and transformations must only include each edit
operation at most once (see Algorithm~\ref{algo:infer-transfo} below). Since
both schemas are finite, these restrictions guarantee that the set of distinct
edit operations is finite and thus, each transformation is finite as well.

With these restrictions, one can view a meta-transformation as inducing a directed acyclic graph over the set of
generated edit operations, where each path is a valid transformation
sequence.

Without the restrictions imposed by \mnextrel, every ordering of the edit
operations is admissible: the induced DAG has a single sink, and all maximal
paths produce the same schema, since they apply the same set of edits. In general,
however, \mnextrel forbids certain orderings, so the DAG has several sinks and
thus several non-equivalent transformations, each producing a distinct schema.
For instance, a set of \mnextrel rules that forbids any \RmvLabelsf\, from being
followed by a further edit produces one schema for each label of the source schema,
each missing a distinct label.

\begin{example}[Deriving and grounding a meta-transformation]
A simple meta-transformation is shown at the top of
Figure~\ref{fig:meta-schema}, in green.
It specifies that a sequence must start with a target update
$\UpdateEdgeTgtsf$ and may be followed only by an object removal
$\RmvObjectsf$. Its evaluation 
over the source schema
in Figure~\ref{fig:meta-graph} proceeds in three steps. First, 
grounding its meta-edit operations 
produces the edit operations $\UpdateEdgeTgtsf(\mathit{hasModerator},\mathit{User})$ and $\RmvObjectsf(\mathit{Moderator})$. In Figure~\ref{fig:meta-graph}, the node $\mathit{Moderator}$ is
removed and the edge $\mathit{hasModerator}$ is redirected. Second, evaluating the $\mstartrelsf$ rule produces:
\[
\startrelsf =
\{\UpdateEdgeTgtsf(\mathit{hasModerator},\mathit{User})\}.
\]
Third, the $\mnextrelsf$ rule produces:
\[
\nextrelsf(
\UpdateEdgeTgtsf(\mathit{hasModerator},\mathit{User})
)
=
\{\RmvObjectsf(\mathit{Moderator})\}.
\]

Hence, the unique maximal admissible transformation sequence is
\[
\RmvObjectsf(\mathit{Moderator})
\circ
\UpdateEdgeTgtsf(\mathit{hasModerator},\mathit{User}).
\]
This respects the ordering constraints
induced by $\startrelsf$ and $\nextrelsf$.
\end{example}

{\small
\begin{algorithm}
  \SetAlgoLined
  \DontPrintSemicolon
  \SetKwComment{tcc}{$\triangleright$}{}
  \SetKwComment{tcc*}{$\triangleright$}{}
  \KwIn{\startrel: the set of starting edit operations, \nextrel: sets of valid succeeding edit operations}
  \KwOut{$\transfoset$: a set of transformations}
  \AlFnt{\Small}
  \caption{\texttt{Derive-Transform}\label{algo:infer-transfo}}

  $\transfoset \leftarrow \emptyset, \transfoset_{t} \leftarrow \{(e_{0}) \mid e_{0} \in \startrel\}$\;

  \While{$\transfoset_{t} \neq \emptyset$}{
    $\transfoset_{t+1} \leftarrow \emptyset$\;

    \ForEach{$(e_{t}\circ\dots\circ e_{0}) \in \transfoset_{t}$}{
      \ForEach{$e \in \nextrel(e_{t})$}{
        \lIf{$e \notin \{e_{0},\dots,e_{t}\}$}{
          $\transfoset_{t+1} \leftarrow
            \transfoset_{t+1} \cup \{e \circ e_{t}\circ\dots\circ e_{0}\}$
        }
      }
    }

    $\transfoset \leftarrow \transfoset \cup \transfoset_{t}$, $\transfoset_{t} \leftarrow \transfoset_{t+1}$\;
  }

  \Return \transfoset\;
\end{algorithm}}

Algorithm~\ref{algo:infer-transfo} enumerates all transformation sequences satisfying the $\startrelsf$ and $\nextrelsf$ constraints. It initializes with single-edit sequences from $\startrelsf$ (line 1) and iteratively extends them with valid successors (line 5), avoiding repeated edits (line 6), until no further extension is possible (line 2), thus producing all duplication-free valid transformations, as established below.

\begin{theorem}[Exhaustive Duplication-free Meta-transformation Derivation]\label{thm:infer-sound}
Algorithm~\ref{algo:infer-transfo} enumerates the exact
duplication-free transformation sequences that start with an edit
operation in $\startrel$ and respect the $\nextrel$ relation.
\end{theorem}

\begin{proof}[Proof Sketch]
Initialization produces exactly the edits in $\startrel$.
Each extension adds only edits in $\nextrel$ and rejects duplicates,
so every generated sequence satisfies the constraints. Conversely, any valid duplication-free sequence consistent with
$\startrel$ and $\nextrel$ can be constructed by successive valid
extensions, which the algorithm enumerates. When no extension is
possible, enumeration stops.
\qedhere
\end{proof}

Theorem~\ref{thm:infer-sound} guarantees exhaustive and duplication-free enumeration with respect to the induced ordering constraints. Next, we address the size of the search space induced by this enumeration and prove that, under finite input schemas, Algorithm~\ref{algo:infer-transfo} produces only finitely many transformations and terminates.

\begin{theorem}[Termination and Finiteness of Transformation Enumeration]
\label{thm:infer-finite}
For any finite input schemas, Algorithm~\ref{algo:infer-transfo}
terminates and produces a finite set of transformation sequences.
\end{theorem}

\begin{proof}
For $t \geq 1$, let $\transfoset_t$ denote the set of
transformations of length $t$ generated by the algorithm.

\emph{Base case.} Since the input schemas are finite,
the meta-relation $\startrel$ is finite.
Each element of $\startrel$ produces
exactly one transformation of length~1
(line~2 of the algorithm).
Hence, $\transfoset_1$ is finite.

\emph{Inductive step.}
Assume that $\transfoset_t$ is finite
for a $t \geq 1$.
Let $\atransfo \in \transfoset_t$
be a transformation with last edit operation $e$.
As input schemas are finite, there is a finite number of distinct edit operations. Indeed, we only consider values from  input schemas when saturating a meta-edit operation. Thus, $\nextrel(e)$ is finite.

Let $n$ denote the total number of distinct edit operations.
By construction, transformations do not contain
duplicate edit operations.
Therefore, from a transformation of length $t$,
at most $n - t$ successor edits can be appended.
Thus, each element of $\transfoset_t$
generates finitely many extensions,
and consequently $\transfoset_{t+1}$ is finite.

When $t = n$, every transformation in $\transfoset_n$
already contains all possible edit operations.
Hence, no successor edit can be appended,
and $\transfoset_{n+1} = \emptyset$.
The loop stops after at most $n$ iterations. Since each of the sets 
$\transfoset_1, \ldots, \transfoset_n$
is finite,
$\transfoset = \bigcup_{t=1}^{n} \transfoset_t$
is finite.
\end{proof}

Let $\editset$ be a set of pairwise compatible edit operations, i.e., any two
operations can be applied in either order. Rather than enumerating all $t!$
permutations, we enumerate subsets of $\editset$ and select a representative
ordering per subset, reducing the search space to $2^{t}$. We model
compatibility as an undirected graph: vertices are edit operations, edges connect
compatible pairs, and pairwise compatible sets are cliques, computed using the
Bron--Kerbosch algorithm~\cite{bron1973algorithm}. We refer to this process as
\emph{clique compression.}

\begin{example}[Commutative relabeling]
Note that, by applying the meta-transformations in
Figure~\ref{fig:meta-schema} we obtain:
\begin{align*}
\startrel = \{&
\Relabelsf(\mathit{Place},\mathit{Location}), \Relabelsf(\mathit{User},\mathit{Person})\} \\
\nextrel = \{&
(\Relabelsf(\mathit{Place},\mathit{Location}),
 \Relabelsf(\mathit{User},\mathit{Person})), \\
&
(\Relabelsf(\mathit{User},\mathit{Person}),
 \Relabelsf(\mathit{Place},\mathit{Location}))
\}.
\end{align*}

The two transformation sequences are equivalent: relabeling $\mathit{User}$
followed by $\mathit{Place}$, and relabeling $\mathit{Place}$ followed by
$\mathit{User}$. In such cases where ordering does not influence the result, it
suffices to keep one.
\end{example}

\paragraph{Meta-Transformation Derivation Complexity.}
Let $n$ be the number of distinct edit operations.
In the worst case, for all distinct $e, e'$, both $(e,e')$ and $(e',e) \in \nextrel$.
Algorithm~\ref{algo:infer-transfo} then enumerates all permutations in $O(n!)$ time. However, since any permutation produces the same schema, this is redundant. By grouping compatible operations into
cliques and enumerating subsets instead of permutations, the worst-case
complexity becomes $O(2^n)$. Section \ref{ssec:graft} presents an efficient algorithm to explore this exponential search space. 

\subsubsection{Transformation generation and execution.}~\label{ssec:validation} Applying a transformation to a schema amounts to sequentially applying its
edit operations to a source schema. Defining transformations as
sequences of atomic edit operations has the advantage that validity checks can
be performed locally and efficiently, since the set of operations is finite and
their effects are precisely specified.

Rather than enumerating validity conditions for each operation, we formulate a generic verification principle. Each edit is applied to a copy of the current schema and is defined by three disjoint sets of schema elements (nodes, edges, labels, or properties): $D$, the elements present before but not after the operation; $I$, the elements absent before but present after; and $C$, the preserved elements. Given a schema $\aschema$, validity requires that $I$, $D$, and $C$ are pairwise disjoint, $I \cap \aschema = \emptyset$, and $D \cup C \subseteq \aschema$. For example, \texttt{\AddLabelsf(o,l)} is characterized by
$C=\{o\}$, $D=\emptyset$, and $I=\{l\}$. Validity requires that $o$
exists in the schema and does not already have the label $l$. Similar
operation-specific checks refine this general principle. The pseudo-code for this algorithm is given in~\cite[App. C]{GRAFTAppendix}.

\paragraph{Schema Transformation Complexity.} Let $n$ be the 
number of distinct grounded edit operations.
With clique compression, Algorithm~\ref{algo:infer-transfo}
generates at most $2^n$ duplication-free transformations.
Each transformation has length at most $n$ and is applied once,
leading to a worst-case running time of $O(2^n \cdot n)$.

\subsection{The Meta-Graph Search Space}\label{ssec:exploration}

Next, we describe how meta-transformations induce a meta-graph search space (Section~\ref{ssec:meta-graph}) whose vertices are schemas
and whose edges are transformation applications. Although
finite, this space can be exponentially large, making exhaustive
exploration infeasible. Schema evolution thus becomes a guided search
problem, raising challenges of reachability (Section~\ref{sssec:reachability}), schema deduplication (Section~\ref{sssec:deduplication}), similarity
computation (Section~\ref{sssec:similarity}), and search space reduction (Section~\ref{sssec:reduction}), which we address below.

\subsubsection{Meta-Graph Construction}\label{ssec:meta-graph}

Note that meta-transformations are schema-agnostic by construction: once grounded, they can be applied to any schema whose symbols satisfy the corresponding instantiation constraints. Consequently, a meta-transformation does not merely encode a single concrete evolution step, but induces a family of schema transitions across different schemas.

This observation gives rise to a binary reachability relation between schemas: a schema $S_2$ is related to a schema $S_1$ if $S_2$ can be obtained by applying a grounded instance of a meta-transformation to $S_1$. The closure of this relation defines a directed search space, which we represent as the \emph{meta-graph} defined below.

\begin{definition}[Meta-Graph]
    Let $\meta{\schemaset}$ be a set of meta-schemas and $\meta{\transfoset}$ be
    a set of meta-transformations. We define the meta-graph obtained from these
    two sets as the directed graph built as follows:

    Its vertex set is the closure of $\meta{\schemaset}$ under the transformations from
   $\meta{\transfoset}$. We then add an arc from a meta-schema \meta{\schema{1}} to a
    meta-schema \meta{\schema{2}} if there exists a meta-transformation $\atransfo
    \in \meta\transfoset$ such that $\meta\atransfo(\meta{\schema{1}})=\meta{\schema{2}}$.
\end{definition}

Figure~\ref{fig:abstract-meta-graph} shows
an example of such a meta-graph between two schemas. The arc contains the
transformation applied to the schema.

\begin{lemma}[Meta-Graph Finiteness]\label{lem:finite}
Let $\mathcal{S}$ be a finite set of source schemas and
$\meta{\transfoset}$ a finite set of meta-transformations.
Then the meta-graph induced by $\meta{\transfoset}$  over $\mathcal{S}$
is finite.
\end{lemma}

\begin{proof}
Let $\mathcal{S}$ be a finite set of schemas, $\meta{\transfoset}$  be a finite set of meta-transformations, and $\tset$ be a finite set of schema elements (nodes, edges, labels, and properties) from at least one schema in $\mathcal{S}$.

By Theorem~1, grounding any meta-transformation in $\meta{\transfoset}$   over schemas in $\mathcal{S}$ produces only finitely many edit operations, since variables range exclusively over elements of $\tset$. Moreover, edit operations cannot introduce fresh symbols outside $\tset$. Hence, every schema reachable by applying meta-transformations is only built with elements of $\tset$. Since $\tset$ is finite, the number of distinct schemas that can be formed from its elements is finite. Hence, the closure of $\mathcal{S}$ under the transformations in $\meta{\transfoset}$ is finite. As the meta-graph contains exactly these schemas as vertices and finitely many transformation edges between them, the meta-graph is finite.
\end{proof}

\paragraph{Meta-graph size.}
Let $n$ be the number of distinct edit operations generated by the
meta-transformations in $\meta{\transfoset}$. As transformations cannot introduce fresh
schema elements, every reachable schema is determined by a subset
of these $n$ operations applied to the source schema. Different
permutations of the same operations either produce the same schema
or are invalid. Hence, the number of distinct reachable
schemas is bounded by $2^n$. The induced meta-graph thus
contains at most $2^n$ vertices and at most
$|\meta{\transfoset}| \cdot 2^n$ edges.

\begin{figure}[t]
\centering
\footnotesize
\begin{tikzpicture}[
    >=stealth,
    metanode/.style={
        circle,
        draw,
        minimum size=8.5mm,
        inner sep=1.5pt,
        align=center
    },
    metaedge/.style={
        ->,
        thick
    },
    metapath/.style={
        ->,
        ultra thick,
        red!70
    }
]

\newcommand{\edgefont}{\scriptsize}

\newcommand{\Relbl}{\mathsf{ReLbl}}
\newcommand{\UpdTgt}{\mathsf{UpdTgt}}
\newcommand{\Rem}{\mathsf{Rem}}

\node[metanode] (src) at (-0.2,0)
  {$\aschema$\\[-2pt]\scriptsize 0.7};

\node[metanode] (s2) at (2.0,0)
  {$\schema{2}$\\[-2pt]\scriptsize 0.75};

\node[metanode] (s3) at (4.8,0)
  {$\schema{3}$\\[-2pt]\scriptsize 0.9};

\node[metanode] (s4) at (4.8,1.1)
  {$\schema{4}$\\[-2pt]\scriptsize 0.77};

\node[metanode] (s1) at (2.0,-1.8)
  {$\schema{1}$\\[-2pt]\scriptsize 0.8};

\node[metanode] (tgt) at (7.2,-1.8)
  {$\target$\\[-2pt]\scriptsize 1.0};

\draw[metaedge] (src) --
  node[midway, above, fill=white, inner sep=0.6pt]
  {{\edgefont$\Relbl(v_P,v_L)$}}
  (s2);

\draw[metaedge] (s2) --
  node[midway, above, fill=white, inner sep=0.6pt]
  {{\edgefont$\Relbl(v_U,v_{Pe})$}}
  (s3);

\draw[metaedge] (s3) --
  node[midway, right=0.1, fill=white, inner sep=0.6pt]
  {{\edgefont$\Relbl(v_X,v_Y)$}}
  (s4);

\draw[metaedge] (s3) --
  node[midway, sloped, above, fill=white, inner sep=0.6pt, align=center]
  {{\edgefont$\UpdTgt(e_{hM},v_U)$}\\[-1pt]
   {\edgefont$\circ\,\Rem(v_M)$}}
  (tgt);

\draw[metaedge] (s2) --
  node[midway, sloped, above, fill=white, inner sep=0.6pt, align=center]
  {{\edgefont$\Relbl(v_U,v_{Pe})$}\\[-1pt]
   {\edgefont$\circ\,\Relbl(v_X,v_Y)$}}
  (s4);

\draw[metapath] (src) --
  node[midway, sloped, below, fill=white, inner sep=0.6pt, align=center]
  {{\edgefont$\UpdTgt(e_{hM},v_U)$}\\[-1pt]
   {\edgefont$\circ\,\Rem(v_M)$}}
  (s1);

\draw[metapath] (s1) --
  node[midway, below, fill=white, inner sep=0.6pt, align=center]
  {{\edgefont$\Relbl(v_P,v_L)$}\\[-1pt]
   {\edgefont$\circ\,\Relbl(v_U,v_{Pe})$}}
  (tgt);

\end{tikzpicture}

\caption{Meta-graph of schema transformations.
Nodes are schemas (similarity to the target shown below nodes).
The highlighted path is the shortest transformation sequence.}
\label{fig:abstract-meta-graph}
\end{figure}
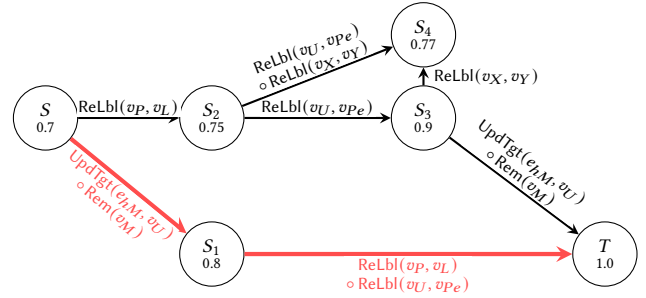

\subsubsection{Reaching the target schema}\label{sssec:reachability}

Let \aschema be a source schema, \target a target schema, and $\meta{\transfoset}$ a set
of meta-transformations. Depending on the expressive power of $\meta{\transfoset}$, the
meta-graph induced by these transformations may not contain a path from
\aschema to \target. For instance, if no meta-transformation in $\meta{\transfoset}$ can
introduce new nodes, then any target schema requiring additional nodes is
unreachable. We therefore assume a set $\meta{\transfoset}$ for which \target is reachable from
\aschema.

\begin{definition}[Target Schema Reachability]\label{def:reachability}
We say that \target is \emph{reachable} from \aschema using $\meta{\transfoset}$ if and
only if there exists a sequence of transformations
$\transfo{1}, \transfo{2}, \dots$ such that
$(\dots \circ \transfo{2} \circ \transfo{1})(\aschema) = \target$.
\end{definition}

If \target is reachable from \aschema using $\meta{\transfoset}$, then an exact transformation from \aschema to \target can be constructed. Otherwise, the best transformation induced by $\meta{\transfoset}$ produces a schema $\simtarget$ that approximates \target, where the degree of approximation depends on the expressive power of the available meta-transformations. In this setting, our goal is to reach such a schema $\simtarget$; when \target is reachable, we obtain $\simtarget = \target$.

\subsubsection{Removing duplicated schemas}
\label{sssec:deduplication}

Applying different transformations to different schemas may yield identical
resulting schemas. For example, two schemas that differ only by the presence of
a single node become identical after applying a transformation that removes that
node. During meta-graph exploration, it is therefore essential to avoid
representing the same schema multiple times.

Detecting schema equivalence amounts to a graph isomorphism problem, which is
computationally expensive in general; moreover, related problems such as
subgraph isomorphism are NP-complete~\cite{DBLP:journals/jacm/Ullmann76}. To avoid
this overhead, we exploit the fact that schema elements are uniquely named.

Specifically, we associate each schema with a canonical \emph{signature}. For
each node, we construct a key consisting of its name, followed by its labels and
properties, both sorted lexicographically. Edge keys are constructed
analogously, incorporating the names of the source and target
nodes. The schema signature is obtained by concatenating all node and edge keys
and sorting them lexicographically.

This signature is invariant under representation choices and uniquely
identifies a schema. Consequently, two schemas are equal if and only if they
share the same signature.

\subsubsection{Evaluating closeness to the target schema}\label{sssec:similarity}

Assessing transformation quality requires a similarity measure between a schema and the target, valued in $[0,1]$, where $1$
denotes identical schemas and $0$ denotes no common structure. A natural candidate, graph edit distance (GED), is impractical here, as computing exact GED is NP-hard even for moderately sized graphs \cite{DBLP:journals/pvldb/ZengTWFZ09}.

We adopt a feature-based measure and decompose each schema into a finite
feature set: (i) node, edge, property, and label names; (ii) labels and
properties on the same node, with its name; (iii) edge and endpoint node names; (iv) labels and properties on a node and its incident edges; and (v)
labels and properties on adjacent nodes.

\begin{theorem}[Exactness]\label{thm:exact}
Jaccard similarity is exact on extracted features: for any schemas
\schema{1}, \schema{2}, $\simil{\schema{1}}{\schema{2}}=1 \Longleftrightarrow \schema{1}=\schema{2}$.
\end{theorem}
\begin{proof}[Proof sketch]
Schema names are unique within their namespace, so feature extraction is
injective up to schema equality. Equal schemas produce equal feature sets, and any difference, e.g., names, type property, type label pair, or edge endpoint, produces a
feature in one set but not the other. 
The full proof is
in~\cite[App. B]{GRAFTAppendix}.
\end{proof}

Given schemas with feature sets $A$ and $B$, we compute their Jaccard similarity as 
$\mathit{sim}(A,B)=\frac{|A\cap B|}{|A\cup B|}$.
Although standard Jaccard is exact, it is biased by the combinatorial features
of types (ii)--(v). Nodes with many properties dominate the feature space,
dampening sensitivity to property variations on sparsely attributed nodes. We
mitigate this with a weighted Jaccard similarity,
$\frac{\sum_i\min(x_i,y_i)}{\sum_i\max(x_i,y_i)}$, where $x_i$ and $y_i$ are the
weights of the $i$-th element of $A\cup B$ in $A$ and $B$ ($0$ if absent). This
preserves exactness while allowing infrequent  features of type (i) to be upweighted. For large sets, MinHash with a parameterized sample size gives an efficient
probabilistic estimate, exact in the limit of large samples. Unless stated
otherwise, we use exact Jaccard.
In practice, any normalizable similarity measure can be used to guide the search. It should, however, range over the structure-aware features above rather than labels
alone. As schemas can share labels but differ structurally, a label-level semantic
distance~\cite{DBLP:journals/csur/ChandrasekaranM21} can disambiguate relabelings but cannot replace
structural context. For performance, a good tradeoff is required between speed of computation, precision (exact is better) and granularity (distinction between close but different schemas).

\subsubsection{Guided Search and Space Reduction}\label{sssec:reduction}

Computing the complete meta-graph is infeasible due to its exponential growth in
both runtime and storage. We hence construct it incrementally and stop the
exploration once a predefined similarity threshold is reached or the similarity
of the best schema found so far no longer improves.

At each step, we expand a subset of the schemas whose neighbors under the available
meta-transformations have not yet been explored. The choice of which schema to
expand is governed by a \emph{strategy}. We consider four strategies.
The \texttt{random} strategy selects a schema uniformly at random and serves as a
baseline. The \texttt{greedy} strategy selects the schema with maximum similarity
to the target. The \texttt{weighted\_distance} strategy favors schemas that
minimize the number of transformations from the source while maximizing
similarity to the target. Finally, the \texttt{naive} strategy expands all
non-expanded schemas in a breadth-first manner. While this guarantees reaching the target schema whenever it is reachable and produces a shortest transformation sequence, it causes an exponential blow-up in the number of generated schemas and is thus impractical.

Applying a transformation may produce exponentially many successor
schemas. Since we evolve toward a fixed target schema, we assume that productive
transformations increase similarity by at least a threshold~$\epsilon$.
To further bound the search space, after each expansion we retain only the $k$
schemas most similar to the target, which we call \emph{candidate schemas}.
The parameter $k$ induces a trade-off, as small values risk pruning paths 
to the target, while large values increase the size of the explored space.

We augment the non-exhaustive strategies (\texttt{greedy}, \texttt{random}, and \texttt{weighted\_distance}) with a beam-width parameter $\beam$, i.e., the number of schemas expanded per iteration. Each strategy selects these $\beam$ schemas by its heuristic: \texttt{greedy} and \texttt{weighted\_distance} keep the top-$\beam$ by their objective, while \texttt{random} samples $\beam$ uniformly.

\subsubsection{An a posteriori cost model for transformations}\label{sec:costmodel}
The  above search only optimizes similarity to \target, treating edits as
cost-free. Since applying a transformation rewrites the underlying instance, we
associate each edit with a data-level cost. Let \propgraph be a property graph
conforming to a schema \aschema, $|\propgraph|$ its number of elements (nodes and
edges), and $m_{\propgraph}(o)$ the number of elements of \propgraph conforming to
a schema object~$o$. For an edit operation~$e$, let $\mathrm{aff}(e)$ be the set of
schema objects whose conforming elements are inserted or deleted by~$e$; for a
node-type deletion, $\mathrm{aff}(e)$ also accounts for the incident edges removed
with the node. The cost of~$e$ is the affected fraction of the graph:
$\mathrm{cost}_{\propgraph}(e)=\frac{1}{|\propgraph|}
  \sum_{o\in\mathrm{aff}(e)} m_{\propgraph}(o)$,
and the cost of a transformation is the sum of its edit-operation costs. A
cost-aware strategy ordering candidate schemas by a weighted combination of
similarity to \target and cost would integrate directly into the framework. A minimum-cost path to
\simtarget could then be extracted by running Dijkstra's algorithm on the meta-graph. Note that a larger tolerance~$\theta$ may be needed so that search terminates before
reaching \target. However, as shown in~\cite[App. E]{GRAFTAppendix}, per-edit cost is
not statically predictable and depends on prior edits (e.g., deleting an edge endpoint removes the edge).


\subsection{The GRAFT Algorithm}\label{ssec:graft}

We now present \graft, a similarity-guided exploration algorithm that incrementally constructs the meta-graph induced by meta-transformations. Rather than fully materializing the search space, \graft explores it with pruning, progressively constructing transformation sequences that increase similarity to the target schema.

\begin{algorithm}[!htbp]
  \SetAlgoLined
  \DontPrintSemicolon
  \SetAlgoNoEnd
  \SetKwComment{tcc}{$\triangleright$}{}
  \SetKwComment{tcc*}{$\triangleright$}{}

  \KwIn{\schemaset: schema set, \meta{\transfoset}: meta-transformation set,
        \target: target schema, $\theta$: threshold, $k$:  |candidate schemas|, {$n_{max}$: run limit, $\beam$: |expanded schemas|}}
  \KwOut{$\mathcal{U}$: transformation set, \simtarget: approx. target schema}

  \AlFnt{\Small}
  \caption{\graft Algorithm}
  \label{algo:graft}

  $G \leftarrow$ empty meta-graph\;
  $\schemaset_{NEW} \leftarrow \schemaset$, $\simtarget \leftarrow \emptyset$, {$n \leftarrow 0$}\;

  \ForEach{$\aschema \in \schemaset$}{
    \lIf{$\simtarget = \emptyset$ \textbf{or} $||\target-\aschema|| {<} ||\target-\simtarget||$}{
      $\simtarget \leftarrow \aschema$
    }
    $G \leftarrow G \cup \{\aschema\}$\;
  }

  \While{$\schemaset_{NEW} \neq \emptyset$ \textbf{and} $||\target-\simtarget|| > \theta$ {\textbf{and} $n < n_{max}$}}{
    $\aschema \leftarrow \mathtt{select}(\schemaset_{NEW}{, \beam})$\tcc*{strategy specific}
    $\schemaset_{NEW} \leftarrow \schemaset_{NEW} \setminus \{\aschema\}$, {$n \leftarrow n+1$}\;

    \ForEach{$\meta{\atransfo}=({\meditset},\mstartrel,\mnextrel) \in \meta{\transfoset}$}{
      $\schemaset' \leftarrow 
        \mathtt{Transform\text{-}Schema}({\meditset},\mstartrel,\mnextrel,\aschema,\target)$\;
      sort $\schemaset'$ by decreasing similarity to $\target$\;
      $\schemaset' \leftarrow$ $k$ first elements of $\schemaset'$\;

      \ForEach{$(\atransfo,\aschema') \in \schemaset'$}{
        \lIf{$||\target-\aschema'|| {<} ||\target-\simtarget||$}{
          $\simtarget \leftarrow \aschema'$, {$n\leftarrow 0$}
        }
        \If{$\aschema' \notin G$}{
          $G \leftarrow G \cup \{\aschema'\}$, $\schemaset_{NEW} \leftarrow \schemaset_{NEW} \cup \{\aschema'\}$\;
          $G \leftarrow G \cup \{\text{edge}(\aschema,\aschema') \text{ with label } \atransfo\}$\;
        }
      }
    }
  }

  ${\transfo{1},\dots,\transfo{j}} \leftarrow 
    \mathtt{path}(G,\simtarget)$\tcc*{shortest path}
  {$\mathcal{U} \leftarrow (\transfo{j}\circ\dots\circ\transfo{1})$}\;

  \Return $\mathcal{U}$\;
\end{algorithm}

Algorithm~\ref{algo:graft} initializes the meta-graph $G$ and the exploration frontier $\schemaset_{NEW}$ with the input schemas (lines~1--5), while tracking the best approximation $\simtarget$ of the target schema. At each iteration (line~6), $\beam$ schemas $\aschema$ are selected from $\schemaset_{NEW}$ according to a strategy (line~7) and expanded by grounding each meta-transformation (lines~9--10). The resulting candidate schemas are ranked by similarity to the target and pruned to the top-$k$ most promising ones (lines~11--12). Newly generated ones are inserted into the meta-graph only if unseen, and edges are labeled with the corresponding transformations (lines~15--17). The process ends when no further expansion is possible, the similarity threshold is reached, or $n_{max}$ successive iterations do not improve the global optimum. Finally, the shortest path from a source schema to $\simtarget$ is extracted (line~18), and the transformation sequence $\mathcal{U}$ is returned (line~19). 

\begin{example}[Similarity-Guided Meta-Graph Exploration]
Consider the meta-graph of Figure~\ref{fig:abstract-meta-graph} with $k=\infty$
and $\theta=0$. Similarities to $\target$ are as labeled:
$\mathit{sim}(\aschema,\target)=0.7$, 
$\mathit{sim}(\schema{1},\target)=0.8$,
$\mathit{sim}(\schema{2},\target)=0.75$,
$\mathit{sim}(\schema{3},\target)=0.9$,
$\mathit{sim}(\schema{4},\target)=0.77$.
Assume the \texttt{select} order $\aschema, \schema{2}, \schema{3}, \schema{1}$. Initially, $G=\{\aschema\}$, $\schemaset_{NEW}=\{\aschema\}$, and
$\simtarget=\aschema$.

Selecting $\aschema$ produces $\schema{1}$ and $\schema{2}$.
Since $0.8>0.7$, $\simtarget$ becomes $\schema{1}$ and
$\schemaset_{NEW}=\{\schema{1},\schema{2}\}$. Selecting $\schema{2}$ produces $\schema{3}$ and $\schema{4}$.
Since $0.9>0.8$, $\simtarget$ becomes $\schema{3}$ and
$\schemaset_{NEW}=\{\schema{1},\schema{3},\schema{4}\}$. Selecting $\schema{3}$ produces $\target$ (and $\schema{4}$, already in $G$).
Since $1.0>0.9$, $\simtarget=\target$ and the search stops.

As $\schema{1}$ was never expanded, no edge $\schema{1}\to\target$
is added to $G$. Hence the shortest path from $\aschema$ to
$\simtarget=\target$ is
$\aschema\to\schema{2}\to\schema{3}\to\target$,
and $\mathcal{U}$ is the composition of the edit-operation sequences
along this path.

This order is purposely inefficient. With the \texttt{greedy} strategy,
iteration~2 selects $\schema{1}$ (0.8$>$0.75), leading to the shorter
path $\aschema\to\schema{1}\to\target$.
\end{example}

\begin{theorem}[GRAFT Termination]\label{thm:graft-termination}
Let $\mathcal{S}$ be a finite set of schemas and $\meta{\transfoset}$ a finite set of 
meta-transformations whose variables range over elements of 
$\mathcal{S} \cup \{T\}$. Then Algorithm~\ref{algo:graft} terminates.
\end{theorem}

\begin{proof}[Proof sketch]
Since meta-transformation variables range only over elements of the
finite input schemas and target schema, grounding produces finitely
many distinct edit operations. These operations cannot introduce
fresh symbols, so every reachable schema is built from a fixed finite
set of schema elements.

Hence, only finitely many distinct schemas can be generated.
Each schema is explored at most once, as previously seen schemas
are not reinserted into the exploration frontier. Therefore, the
algorithm performs only finitely many expansions and terminates.
\end{proof}

\begin{theorem}[Soundness of GRAFT]
Let \aschema and \target be schemas and \meta{\transfoset} a finite set of meta-transformations whose
variables range over elements of $\aschema \cup \target$.
Algorithm~\ref{algo:graft} is sound: every transformation sequence returned by Algorithm~\ref{algo:graft}
    maps \aschema to a schema reachable from \aschema via grounded instances of
    meta-transformations in \meta{\transfoset}.
\end{theorem}
\begin{proof}
As the reachability relation is transitive,
    if \schema{i} is reachable from \schema{i-1}, any schema generated
    from \schema{i} using \meta{\transfoset} is reachable from \schema{i-1}.
    All schemas produced by Algorithm~\ref{algo:graft} were either obtained
    directly from \aschema or by a sequence of transformations starting at
    \aschema. Thus, all schemas produced are reachable from \aschema. \qedhere
\end{proof}
    
We consider schema losslessness, i.e., transformations that preserve  data conforming to the source schema under the result schema.

\begin{theorem}[GRAFT Losslessness]
Let every grounded meta-transformation applied in Algorithm~\ref{algo:graft}
be lossless at the schema level. Then every transformation returned
by Algorithm~\ref{algo:graft} is lossless.
\end{theorem}

\begin{proof}
Algorithm~\ref{algo:graft} returns a composition of grounded 
meta-transformations, each lossless by assumption. As losslessness 
is closed under composition, the result transformation is lossless.
\end{proof}

Losslessness must be ensured both at the meta-transformation
and grounding levels. Meta-transformations must generate only
edit sequences that do not invalidate conforming data, and
grounding must respect the ordering constraints given by
\textit{Start}$^*$ and \textit{Next}$^*$ so that dependent edits are applied
safely. When exact target reachability is required, all necessary
target elements must be grounded.

\begin{theorem}[GRAFT Minimality]\label{thm:minimality}
If no pruning is applied, the \texttt{naive} strategy is used, and
meta-transformations do not generate redundant or inverse edit
operations, then the transformation sequence returned by
Algorithm~\ref{algo:graft} is minimal in the number of
edit operations.
\end{theorem}

\begin{proof}
Under the \texttt{naive} strategy, Algorithm~\ref{algo:graft} explores
transformations in increasing length order. Hence, the first sequence reaching $T$
is minimal. No shorter equivalent exists, as Algorithm~\ref{algo:infer-transfo} forbids duplicate edits within a sequence, the constraints rule out generating an operation and its inverse, and edits are atomic.
\end{proof}

We further show that pruning does not cost greedy its optimality when the
transformation set is well behaved. We denote by \besttarget the best reachable schema from the source \aschema, that is $\forall \schema{k} \in \reach(\aschema), \simil{\besttarget}{\target} \ge \simil{\schema{k}}{\target}.$

\begin{definition}[\target-monotone meta-transformation set]\label{def:tmonotone}

A set \meta\transfoset of meta-transformations is \emph{\target-monotone} if, for
every $\schema{k}\in\reach(\aschema)$ with
$\simil{\schema{k}}{\target}<\simil{\besttarget}{\target}$, some
meta-transformation $\meta\atransfo\in\meta\transfoset$ admits a grounding
\atransfo\ with $\simil{\atransfo(\schema{k})}{\target}>\simil{\schema{k}}{\target}$.
\end{definition}

\begin{theorem}[Greedy optimality under pruning]\label{thm:greedy-opt}
Let \meta\transfoset be \target-monotone. Then for every $k\ge 1$,
$\mathrm{beam}\ge 1$, $n_{\max}\ge 1$, and $\theta=0$, Algorithm~\ref{algo:graft}
with the greedy heuristic returns an optimal \besttarget. If \target
is reachable from \aschema, it returns an exact transformation to \target.
\end{theorem}
\begin{proof}[Proof sketch]
The current best schema \simtarget always remains in the frontier
$\schemaset_{NEW}$ and, while sub-optimal, attains the maximum frontier similarity,
so greedy selects it; \target-monotonicity then guarantees a strictly more similar
successor, so the best similarity strictly increases. As $\reach(\aschema)$ is
finite (Lemma~\ref{lem:finite}), the optimum is reached after finitely many
iterations. The full proof is in~\cite[App. D]{GRAFTAppendix}.
\end{proof}

Should the meta-transformation set not be \target-monotone, GRAFT can still reach optimality at the cost of disabling pruning and using the exhaustive \texttt{naive} strategy. A proof is given in~\cite[App. I]{GRAFTAppendix}.

\paragraph{Complexity Analysis.}
\graft explores the meta-graph by similarity guided search. Let $d$ be the
minimum number of transformations to reach the target, $b$ be the branching factor
of grounded meta-transformations, and $C$ be the cost of applying one
(Section~\ref{ssec:validation}). Without pruning, \graft explores $O(b^d)$
schemas, running in $O(b^d \cdot C)$. With pruning width $k$, depth $i$ expands
at most $k^i$ schemas into $k^i b$ candidates. Pruning them with a bounded heap adds a $O(\log k)$ factor, leading to a bound of $O(k^d \cdot b \cdot (C + \log k))$.


\section{Experimental Evaluation}\label{sec:evaluation}

We evaluate GRAFT on multiple schema evolution scenarios, measuring heuristic efficiency and scalability, sensitivity to dataset characteristics and operation reuse, and performance breakdown.


\subsection{Quantitative Study}\label{ssec:quantitative}

We assess \graft across multiple heuristics and synthetic datasets.

\subsubsection{Experimental Setup}\label{ssec:setup}

\graft is implemented in Rust, using the Souffl\'e engine \cite{DBLP:conf/cav/JordanSS16} to evaluate transformation rules. We used Neo4j Community Edition 5.24.0 to store the meta-graph and compute the final transformation paths. Benchmarks ran on an Artix Linux machine with an AMD Ryzen 9 5900X and 32 GB of RAM.

\emph{Datasets.} We base our experiments on the schemas of the iBench benchmark~\cite{DBLP:journals/pvldb/ArocenaGCM15}. For each source schema, we specify a set of Datalog meta-transformations mapping it to corresponding target schemas adapted from~\cite{DBLP:journals/pvldb/BonifatiRMFE24}. Their characteristics and those of associated source-target transformations are given in Table~\ref{table:datasets}. We define the branching of a transformation as the maximum number of edit operations that can immediately follow a given operation. Both \texttt{DBLP-Amalgam1} and \texttt{Amalgam1-Amalgam3} originate in schema migration between bibliographic databases. Each step adds a node with a label, a variable number of properties, and a labeled edge. In \texttt{Person-Data}, each step adds a node with a label and one property, and a labeled edge, without validity checks (e.g., a property may be added even if not allowed in the target). In \texttt{Flight-Hotel}, transformations may add a node, an edge, a label to a node, a property to a node, or a label to an edge. {These two model migrations from relational to graph databases.} We report the number of nodes in the complete meta-graph, i.e., the total number of schemas, in Table~\ref{table:datasets}. \texttt{IBench-Huge} was generated using iBench with a target of 500 nodes, producing a schema as large as the largest schemas from the literature \cite{DBLP:conf/icde/AlotaibiLQEO21,DBLP:conf/edbt/BonifatiDM22}. It consists of transformations that merge two nodes, split a node into two, add and remove properties.

\begin{table*}[t]
    \small
    \caption{Dataset characteristics for source and target schemas and their associated meta-transformations.}
    \centering
    \begin{tabular}{l|cccc|cccc|c|ccc|cc}
        \multirow{2}{*}{Dataset}
        & \multicolumn{4}{c|}{Source Schema}
        & \multicolumn{4}{c|}{Target Schema}
        & \multirow{2}{*}{Sim.}
        & \multicolumn{3}{c|}{Source-Target Transfo.}
        & \multirow{2}{*}{\# Meta-Nodes} \\
        & |$\mathcal{V}$| & |$\mathcal{E}$| & |$\mathcal{K}$| & |$\mathcal{L}$|
        & |$\mathcal{V}$| & |$\mathcal{E}$| & |$\mathcal{K}$| & |$\mathcal{L}$|
        & & \# Op. & Branching & \# Transfo. &\\
        \hline
       \texttt{Person-Data (D1)}       & 3  & 0 & 6   & 3  & 6  & 2 & 9   & 8  & 0.41 & 13 & 5  & 2 & 61 731\\
       \texttt{DBLP-Amalgam1 (D2)}     & 7  & 0 & 46  & 8  & 9  & 1 & 64  & 10 & 0.63 & 23 & 16 & 2 & 3 211 329\\
       \texttt{Amalgam1-Amalgam3 (D3)} & 15 & 0 & 101 & 15 & 17 & 1 & 119 & 18 & 0.84 & 24 & 16 & 2 & 3 407 941\\
       \texttt{Flight-Hotel (D4)}      & 2  & 0 & 5   & 2  & 5  & 3 & 9   & 8  & 0.30 & 16 & 0  & 5 & 65 536\\
   {\texttt{IBench-Huge (D5)}}      & {500}  & {1} & {2500}   & {500}  & {500}  & {1} & {2501} & {500}  & {0.98} & {24} & {6}  & {4} & {8 877 932}\\
    \end{tabular}
    \label{table:datasets}
\end{table*}

\emph{Methodology.}
Unless stated otherwise, we use a pruning width of 8 (keeping the eight most similar schemas per round), which gave a good speed–accuracy tradeoff in our experiments. Search stops when the best similarity improves by less than 0.01 over five consecutive steps, or after 6 hours. Similarity defaults to the exact Jaccard index over extracted graph features.


\graft searches offline at the schema level, leaving instance data untouched. Each selected transformation is converted to Cypher via parameterized per-operation templates: deletions translate directly, while insertions draw values from the chosen data source (e.g., copying a moved property), possibly requiring human input for new properties. We use this for the measurements in Section~\ref{ssec:qualitative}.

\subsubsection{Experimental Results}\label{ssec:results}

We evaluate the effectiveness, efficiency, robustness, and generalization of our approach across various heuristics and datasets, addressing \textbf{RQ\ref{rq:pruning}-\ref{rq:reuse}} below. Results on scalability (\textbf{RQA}) and on runtime distribution across computation steps (\textbf{RQB}) are in the extended version~\cite[App. F,H]{GRAFTAppendix}.

\rquest\label{rq:pruning} \emph{How robust are the different search strategies to the number of { retained and expanded} candidate schemas at each step?}

We vary the pruning width (number of candidate schemas kept after each round) and measure (i) the normalized similarity improvement from
source to target schema and (ii) the runtime and transformation size of each
strategy. { We replace \texttt{DBLP-Amalgam1} with the large-scale \texttt{IBench-Huge} schema and add an unbounded setting (``u'') that disables pruning (results for \texttt{DBLP-Amalgam1} are in ~\cite[App. J]{GRAFTAppendix}). We vary beam width as for pruning width, up to 16, large enough for strategy specific patterns to appear. Reaching it requires setting the pruning width to 16 as well, so enough candidates are available for expansion. Figure~\ref{fig:pruning-experiment} reports varying pruning width (top) and beam width (bottom).
}

Figure~\ref{fig:pruning-experiment}a reports the normalized similarity
improvement. \texttt{greedy} and \texttt{naive} consistently reach~1.0 across all settings,
meaning they eventually produce a schema identical to the target; {\texttt{naive} times out at 8 candidate schemas for \texttt{IBench-Huge}.} The final similarity of \texttt{random} 
varies across datasets and pruning
widths. A characteristic pattern appears for \texttt{weighted\_distance}:
it reaches similarity~1.0 only for small pruning width (up to 8), 
converging to a lower value for larger widths. This is consistent with its
objective of balancing similarity and the number of transformations rather than
maximizing similarity alone. { Higher beam widths have the opposite effect since a sub-optimal decision is mitigated.} When few candidates are kept, the “balanced”
solution is pruned early as it is not yet similar enough to the target;
with more candidates, this solution becomes reachable. { Disabling pruning does not affect the precision of \texttt{naive} and \texttt{greedy}, but can grow the candidate pool exponentially. This gives \texttt{random} and \texttt{weighted\_distance} a better chance at finding good solutions, though on \texttt{IBench-Huge} they plateau at 98\%.}


Figure~\ref{fig:pruning-experiment}b shows runtime growing with the number of candidates for all strategies, most steeply for \texttt{naive}. On \texttt{Amalgam1-Amalgam3}, unpruned \texttt{weighted\_distance} also times out, likely because it favors schemas far from the target, which admit many more valid transformations under \mnextrel. On \texttt{Flight-Hotel}, unpruned \texttt{naive} beats width 16: each step produces at most 16 candidates, so pruning only adds overhead.

Figure~\ref{fig:pruning-experiment}c reports the 
resulting transformation length, and the first four plots of
Figure~\ref{fig:dataset-experiment}a report the average size of
the selected transformations. The \texttt{greedy} and
\texttt{naive} heuristics are highly stable: they always pick the same path and
produce transformations of constant size across pruning and beam widths.
As beam width does not affect naive, it is excluded from that last experiment.

Figure~\ref{fig:dataset-experiment}c shows the normalized similarity
improvement for the four datasets when improving the similarity measure. That
is, each value corresponds to the number of features used by the MinHash
approximation of the Jaccard index. We include the Jaccard index itself and the weighted Jaccard variant, in which the only non-unit weights
    are assigned to labels (4), properties (2), and pairs of properties sharing a
node (0.25).
The precision of the similarity measure has little impact on the results for
these datasets. The exception is weighted Jaccard: by giving higher
    weight to rare but important features such as individual properties, it shows a
larger drop than standard Jaccard when a property is missing.

\begin{keytakeaway}[RQ1]
    \texttt{greedy} is the most robust to pruning { and beam width}, consistently reaching the target without timeouts.
\end{keytakeaway}

\begin{figure*}[!htbp]
    \centering
    \includegraphics[width=\textwidth, trim=0 .34cm 0 .4cm, clip]{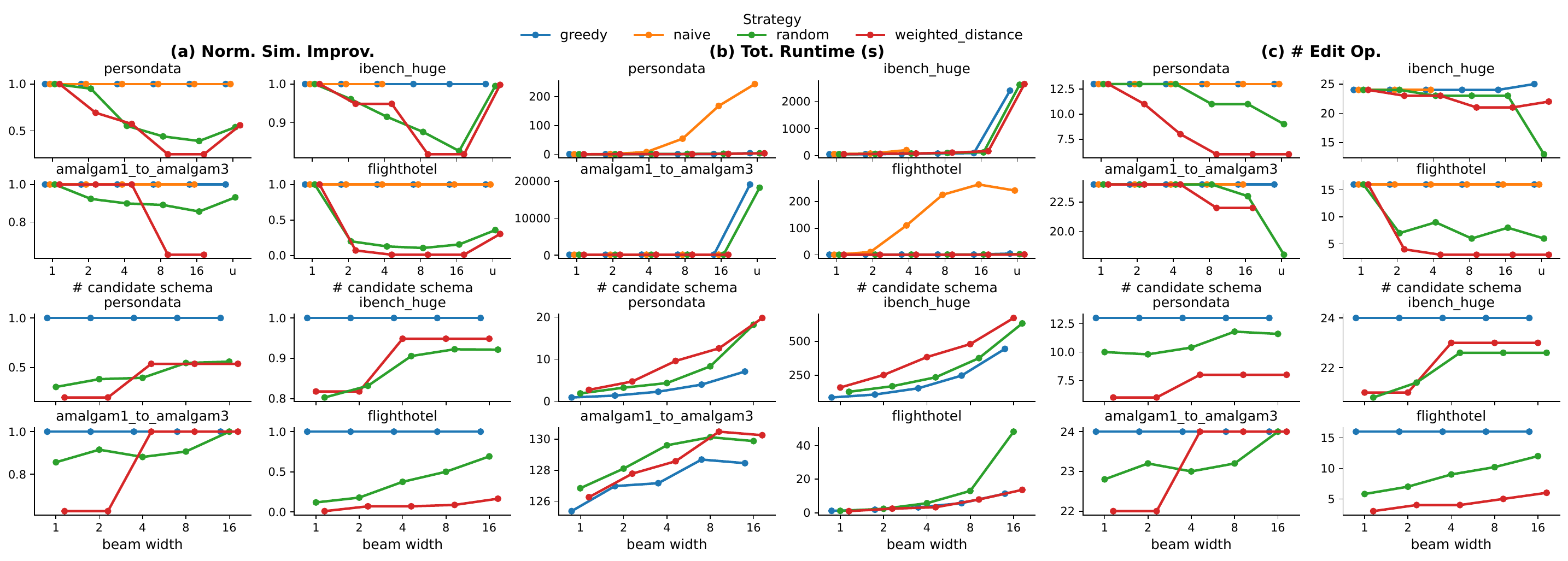}
    \caption{{Y-axis: normalized similarity, runtime, edit operations; X-axis: no. of candidate schemas (top), beam width (bottom).}}\label{fig:pruning-experiment}
\end{figure*}

\begin{figure*}[!htbp]
    \centering
    \includegraphics[width=\linewidth, trim=0 .34cm 0 .43cm, clip]{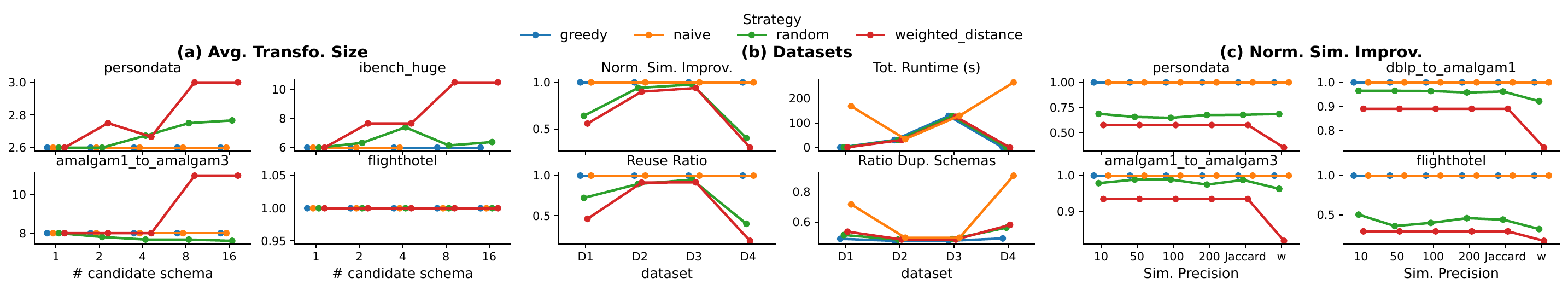}
    \caption{Average transformation size by candidate-schema count (left), strategy performance across datasets (middle, labels in Table~\ref{table:datasets}), and similarity improvement by MinHash sample count for Jaccard and weighted Jaccard{} approximations (right).}
    \label{fig:dataset-experiment}
\end{figure*}


\rquest\label{rq:dataset} \emph{How do the design and constraints of meta-transformations
affect convergence and runtime across datasets?}

Intuitively, more restrictive meta-transformations (fewer applicable edits and
stronger validity checks) should reduce the search space, thus improving both
convergence and runtime. At the same time, the granularity of the
meta-transformations (i.e., the size of the transformations they generate) also
plays a key role.

We apply all strategies to the four datasets, each instantiated with a different
set of meta-transformations, and report the best similarity achieved and the
total runtime. The top left plot of Figure~\ref{fig:dataset-experiment}b shows that the \texttt{greedy} and \texttt{naive} heuristics {always reach the best similarity}. The \texttt{random} and
\texttt{weighted\_distance} heuristics perform comparatively better on
\texttt{Amalgam1-Amalgam3} and \texttt{DBLP-Amalgam1},
but this is largely due to these source schemas being closer to their targets in
the first place, rather than to better search behavior.

Concerning
total runtime (Figure~\ref{fig:dataset-experiment}b, first row, second column), \texttt{naive} is the
only strategy whose behavior differs markedly from the others, because it
explores \emph{all} newly produced schemas at each step. Hence, its runtime
is highly sensitive to the number of schemas generated by the
meta-transformations, unlike \texttt{greedy}, \texttt{random}, and \texttt{weighted\_distance},
which select a single schema per step.

\begin{keytakeaway}[RQ2] Lower branching speeds convergence; all heuristics are branching-insensitive, except \texttt{naive} which scales poorly.
\end{keytakeaway}


\rquest\label{rq:reuse} \emph{To what extent do the different strategies reuse the original
edit operations encoded in the meta-transformations?}

For each dataset and strategy, starting from the original edit sequence used to derive the meta-transformations, we measure the fraction of these edits that appear in the resulting transformation path, i.e., how well the strategy recovers the original mapping.

Figure~\ref{fig:dataset-experiment}b (bottom, left) reports the percentage of original
edit operations reused by each strategy. \texttt{greedy} reaches $100\%$
reuse, producing optimal, shortest edit paths to the target schema containing all original edit
operations.
\texttt{random} and \texttt{weighted\_distance} reach higher reuse with more constrained 
meta-transformations (e.g., \texttt{Person-Data} and
\texttt{DBLP-Amalgam1}), but never full reuse.

\begin{keytakeaway}[RQ3]
\texttt{greedy} consistently recovers all original edit operations;
\texttt{random} and \texttt{weighted\_distance} only partially reuse.
\end{keytakeaway}


\subsection{Qualitative Study}\label{ssec:qualitative}

We assess the quality of \graft's transformations on both real-world (\textbf{RQ\ref{rq:icij}}) and synthetic large-scale data (\textbf{RQ\ref{rq:instance}}).

\begin{figure}[!htbp]
    \centering
    \includegraphics[width=.4\textwidth, trim=1cm .0cm 1cm .0cm, clip]{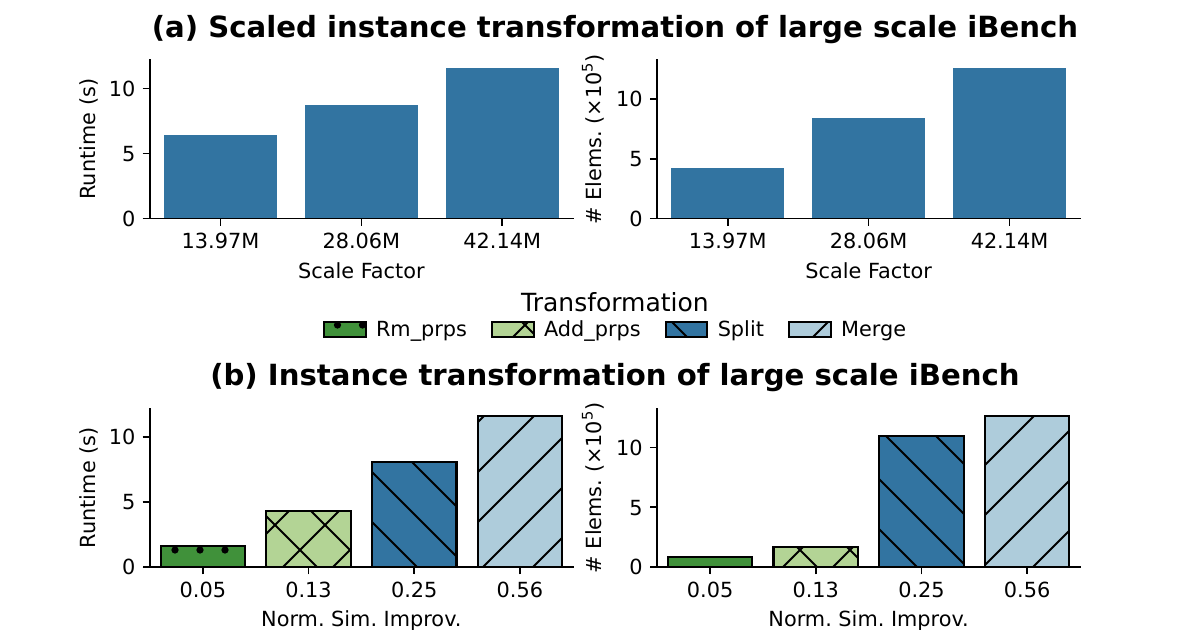}
    \caption{ Data-level runtime and affected elements on \texttt{IBench-Huge} by scale factor (top) and similarity (bottom).
    \label{fig:ibench-result}}
\end{figure}

\rquest\label{rq:icij} \emph{Do the derived meta-transformations lead to similar evolution paths under different heuristics, and can they be reused?}

Our evaluation uses the ICIJ Paradise Papers dataset~\cite{icijdb}, modeling offshore \texttt{Entities}, their \texttt{Officers} (beneficiaries, directors, or shareholders), \texttt{Intermediaries}, and registered \texttt{Addresses}. ICIJ provides source and target schemas plus data-level rules that only add nodes, edges, or properties, so the source schema is a strict subset of the target. We lift these to the schema level, dropping data-dependent conditions and replacing each constant with a variable, which abstracts the 18 original rules into 11 distinct meta-transformations, as several induce identical schema changes. We report the best \texttt{greedy} results with candidate set size 8, since larger values match these but raise runtime when escaping local optima.

Naively applying rules reaches a schema with $73\%$ similarity to the target in
$1.57$ seconds, producing a 10-step transformation sequence. It draws on four 
meta-transformations derived from ICIJ
{(full definitions are provided in~\cite[App. G]{GRAFTAppendix}) that respectively add a node-property pair, two such pairs joined by an
edge, an edge between two existing nodes, and a node with a self-loop.}

The observed similarity results from local optima in the meta-graph, specifically from the self-loop transformation. The \texttt{greedy} heuristic systematically favors such edges, as they contribute more strongly to the similarity measure by encoding both structural connectivity and label-property relationships. Once applied, however, it cannot be further refined, as attaching properties requires reinserting the node, which the default semantics forbid. Hence, the
search becomes trapped in a locally optimal region of the meta-graph.

To assess the impact of this restriction, we repeat the experiment using the
same \texttt{greedy} strategy and candidate set size, but enforce node insertion
idempotency, i.e., the fact that inserting an existing node has no effect. \emph{Under
this modified semantics, the algorithm reaches the target schema exactly,
achieving $100\%$ similarity.} {The computation completes in $8.9$ seconds, 
chooses additional, more complex meta-transformations (detailed in~\cite[App. G]{GRAFTAppendix}), and produces a
transformation sequence of length~32}. The higher runtime is due to the increased
number of edit operations applied.

\begin{keytakeaway}[ICIJ Use-Case]
Data-derived rules risk local optima, which only idempotency or exponential \texttt{naive} search avoids.
\end{keytakeaway}

{
\noindent\rquest\label{rq:instance} \emph{Are the produced transformations applicable in practice?}

To assess practical applicability, we materialized a property graph instance
from the iBench schema. The instance contains $42.14$M nodes with five properties each,  distributed uniformly over the 
$500$ node types ($84.28$K per type), together with $84.28$K edges.

Each node has an indexed \texttt{id}. The number of impacted elements and
the runtime are taken from the Neo4j query summary after each Cypher
statement. For worst-case behavior, we do not index newly inserted
nodes and leave all queries unoptimized.

For each transformation, Figure~\ref{fig:ibench-result}b reports its similarity improvement against runtime and number of affected elements. The latter is cumulative, so transformations can revisit the same nodes once per edit operation. The first query also incurs index-loading overhead.

Both runtime and the number of affected elements 
grow moderately with transformation similarity, as larger
similarity steps entail larger structural changes. 
ICIJ data shows the same trend~\cite[App. G]{GRAFTAppendix},
with runtime affected by node-type skew.

To isolate the effect of
instance size, we repeat the experiment with the most similarity-increasing
transformation, for
$13.97$M, $28.06$M, and $42.14$M nodes.
Figure~\ref{fig:ibench-result}a shows that both
runtime and the number of affected elements grow linearly in the database size.

\begin{keytakeaway}[]
Minimal transformations affect only elements differing between source and target schemas, keeping cost small relative to database size, though operation-dependent.
\end{keytakeaway}
}

\section{Conclusions and Perspectives}\label{sec:conclusion}

We presented \graft, a framework for property graph schema evolution and
reuse. Given a set of schema transformations, \graft derives mappings from one
or more source schemas to a target schema, modeling evolution as a meta-graph
whose vertices are schemas and whose edges are grounded meta-transformations.

Pruning keeps similarity-guided exploration efficient while preserving accuracy.
Without pruning, \graft is complete. Under pruning, 
its greedy heuristic reaches the
target for $T$-monotone sets, i.e., those admitting a similarity-improving step at
every non-optimal schema.
 Since mappings
are compositions of transformations, any property preserved by every step is
inherited by the mapping, ensuring consistency.

\graft goes beyond schema matching by deriving executable and
composable transformation sequences rather than one-off mappings.
Potential applications include schema migration tools, which can consume these
sequences directly. As each sequence constructs the target schema from the source, it produces an alignment that can then be used in ontology and knowledge graph
integration.

Our work opens several directions. Semantic similarity could complement
our exact matching and allow discovering label correspondences, while
value-domain transformations such as unit conversions could be supported as
value-rewriting edit operations. \graft could further integrate property graph
schema constraints through richer validity conditions and dedicated edit
operations. Finally, the cost model could be made dynamic and used to drive the search, using the query plan to predict a
transformation's data-level cost while accounting for dependencies between edits.
This estimate could then order candidate schemas by a weighted combination of
similarity to the target and predicted cost.

\bibliographystyle{ACM-Reference-Format}
\bibliography{biblio}

\clearpage
\appendix

{\section*{Appendix}

This appendix collects material deferred from the main text. 

Appendix~\ref{sec:relational} illustrates the differences between relational and property graph schema evolution on the motivating example.
Appendix~\ref{sec:Jaccard}
establishes \graft's Jaccard similarity measure is exact (Theorem~\ref{thm:exact}).
Appendix~\ref{sec:transformations} provides the schema-transformation algorithm
and proves its termination, and
Appendix~\ref{sec:optimality} proves the optimality of greedy search under
pruning for \target-monotone meta-transformation sets. Appendices~\ref{sec:per-edit-op}-\ref{sec:breakdown}
report additional experiments: operation cost per edit
(Appendix~\ref{sec:per-edit-op}), scalability with schema size (Appendix~\ref{sec:scalability}), a real-world ICIJ case study
(Appendix~\ref{sec:icij-instance}), and a runtime breakdown across pipeline
stages (Appendix~\ref{sec:breakdown}). Appendix~\ref{sec:completeness}
proves the completeness of unpruned \graft.}

\section{Relational vs. Property Graph Schema Evolution ~\label{sec:relational}}

\newcommand{\new}[1]{{#1}}
\newcolumntype{Y}{>{\raggedright\arraybackslash}X}

Figure~\ref{fig:rel-analogue} illustrates how the running example of
Figure~\ref{fig:meta-graph} might be structured in the relational model. Node
types are converted to tables and edge types to foreign keys. In this
setting, edit operations differ in their impact on the schema (see
Table~\ref{tab:rel-vs-graph} for a comparison). The closest relational
analogue to our meta-edit operations is the line of work on
schema-modification operators (SMOs), in particular the PRISM
workbench~\cite{DBLP:journals/pvldb/CurinoMZ08}, where a relational schema
evolves through a fixed catalogue of atomic operators and rewrite queries.

Note, however, that relational SMOs act on tables and attributes that can
largely evolve independently, e.g., adding or dropping a column does not necessarily
force changes elsewhere in the schema. In a property graph schema, structural
and typing constraints are tightly coupled. Removing a node type leaves all
incident edges dangling, while renaming one rewrites the endpoint typing of
every incident edge. Consequently, the evolution scenario of
Figure~\ref{fig:meta-graph} cannot be expressed as a set of independent edits:
\texttt{UpdTgt} on \texttt{hasModerator} must precede the removal of the
\texttt{Moderator} node type, and the relabelings interact with both. These
ordering constraints are captured by \mstartrel{} and \mnextrel{}.

\begin{figure}[!ht]
\centering\small\setlength{\tabcolsep}{3pt}\renewcommand{\arraystretch}{1.1}
\begin{tabular}{@{}>{\raggedright\arraybackslash}p{.47\linewidth}>{\raggedright\arraybackslash}p{.47\linewidth}@{}}
\toprule
\textbf{V1 (source)} & \textbf{V2 (target)}\\
\midrule
\textsf{Forum}(\underline{fid}, title) & \textsf{Forum}(\underline{fid}, title)\\
\textsf{User}(\underline{uid}, firstName, lastName) & \new{\textsf{Person}}(\underline{uid}, firstName, lastName)\\
\textsf{Moderator}(\underline{mid}, firstName, lastName) & \new{(dropped)}\\
\textsf{Place}(\underline{pid}, name) & \new{\textsf{Location}}(\underline{pid}, name)\\
\textsf{hasMember}(fid$\to$\textsf{Forum}, uid$\to$\textsf{User}, creation) & \textsf{hasMember}(fid$\to$\textsf{Forum}, uid$\to$\new{\textsf{Person}}, creation)\\
\textsf{hasModerator}(fid$\to$\textsf{Forum}, \new{mid$\to$\textsf{Moderator}}) & \textsf{hasModerator}(fid$\to$\textsf{Forum}, \new{uid$\to$\textsf{Person}})\\
\textsf{isLocatedIn}(uid$\to$\textsf{User}, pid$\to$\textsf{Place}) & \textsf{isLocatedIn}(uid$\to$\new{\textsf{Person}}, pid$\to$\new{\textsf{Location}})\\
\bottomrule
\end{tabular}

\begin{lstlisting}
RENAME TABLE Place TO Location; RENAME TABLE User TO Person;
ALTER TABLE hasModerator DROP FOREIGN KEY (mid);
ALTER TABLE hasModerator RENAME COLUMN mid TO uid,
  ADD FOREIGN KEY (uid) REFERENCES Person;
DROP TABLE Moderator;
\end{lstlisting}
\caption{Relational analogue of Figure 1.}
\label{fig:rel-analogue}
\end{figure}

\begin{table}[!ht]
\caption{Relational SMOs vs. \graft operations.}
\label{tab:rel-vs-graph}
\centering\small\renewcommand{\arraystretch}{1.2}
\begin{tabularx}{\linewidth}{@{}lYY@{}}
\toprule
& \textbf{Relational} & \textbf{Property graph}\\
\midrule
Rename & \texttt{RENAME TABLE}: local; referencing FKs stay valid. & \texttt{ReLbl}: endpoint types of incident edges change with it.\\
Re-point & Update the FK to a new referenced table. & \texttt{UpdTgt} must be applied before the old endpoint is removed.\\
Drop & \texttt{DROP TABLE} once no FK references it. & \texttt{Rem} leaves incident edges dangling.\\
Order & Explicit, detachable via declared FKs. & Implicit, mandatory (\texttt{UpdTgt}\,$\prec$\,\texttt{Rem}); enforced by \texttt{Start*}/\texttt{Next*}.\\
\bottomrule
\end{tabularx}
\end{table}

\section{Exactness of Jaccard similarity}\label{sec:Jaccard}

We restate here Theorem~\ref{thm:exact}:
\begin{extthm}{\ref{thm:exact}}
On the extracted features, Jaccard similarity is exact for property graph
schemas: for any two schemas \schema{1} and \schema{2},
\[
  \text{sim}\!\left(\schema{1},\schema{2}\right) = 1
  \quad\Longleftrightarrow\quad
  \schema{1} = \schema{2}.
\]
\end{extthm}

\begin{proof}

A schema $\aschema = \left(\mathcal{V}_S,\mathcal{E}_S,\rho_S\right)$ has
disjoint finite sets $\mathcal{V}_S$ of vertex types and $\mathcal{E}_S$ of edge
types, each type named uniquely within $\mathcal{V}_S\cup\mathcal{E}_S$, together
with an endpoint map
$\rho_S:\mathcal{E}_S\rightarrow\mathcal{V}_S\times\mathcal{V}_S$. We write
$\rho_S(e)=\left(\xi(e),\tau(e)\right)$, so that
$\xi,\tau:\mathcal{E}_S\rightarrow\mathcal{V}_S$ return the source and target
of each edge type; these are well defined since each edge type is fixed by its
name. The proof additionally uses the property and label assignments
$\pi:\mathcal{V}_S\cup\mathcal{E}_S\rightarrow \powset{\text{Props}}$ and
$\lambda:\mathcal{V}_S\cup\mathcal{E}_S\rightarrow \powset{\text{Labels}}$ carried by
the schema, whose names may recur across types. Two schemas are equal if and
only if their vertex and edge sets, their endpoint maps, and their property and
label assignments all coincide.

The extractor $F$ keeps names in four separate namespaces, so
a vertex name, an edge name, a property name, and a label name never collide
even when they share the same string. It produces:
\begin{itemize}
  \item \textbf{Names}, i.e., the vertex name of each $t\in\mathcal{V}_S$, the edge name
  of each $t\in\mathcal{E}_S$, every property name in $\bigcup_t\pi(t)$, and every
  label name in $\bigcup_t\lambda(t)$;
  \item \textbf{Incidences}, i.e., for each $t\in\mathcal{V}_S\cup\mathcal{E}_S$, the pairs
  $\left(t,p\right)$ for $p\in\pi(t)$ and $\left(t,\ell\right)$ for
  $\ell\in\lambda(t)$; and for each $e\in\mathcal{E}_S$, the triples
  $\left(e,\xi(e),\text{src}\right)$ and $\left(e,\tau(e),\text{tgt}\right)$.
\end{itemize}
Structural features, such as two properties sharing a type, follow from the
incidences; they do not affect exactness but refine the score when it is
below~$1$.

Jaccard similarity equals $1$ exactly when
$F\!\left(\schema{1}\right)=F\!\left(\schema{2}\right)$, and equal schemas
clearly yield equal feature sets. It remains to prove the converse, which we do
by contrapositive: assuming $\schema{1}\neq\schema{2}$, we exhibit a feature
present in one set but not the other. As this is symmetric in the two schemas,
we may fix the direction by swapping them if needed.

If $\schema{1}\neq\schema{2}$, the components diverge in one of
three ways.

\begin{enumerate}
  \item \emph{The vocabularies differ}, that is, $\mathcal{V}_{S_1}\neq\mathcal{V}_{S_2}$,
  $\mathcal{E}_{S_1}\neq\mathcal{E}_{S_2}$, $\bigcup_t\pi_1(t)\neq\bigcup_t\pi_2(t)$, or
  $\bigcup_t\lambda_1(t)\neq\bigcup_t\lambda_2(t)$. Then some name occurs in
  $\schema{1}$ but not in $\schema{2}$, and its name feature, taken in the
  matching namespace, separates the two feature sets. 
  \item \emph{The vocabularies agree but an attachment differs.} Suppose Case~1
  does not hold and $\pi_1\neq\pi_2$. Since $\mathcal{V}_{S_1}\cup\mathcal{E}_{S_1}
  =\mathcal{V}_{S_2}\cup\mathcal{E}_{S_2}$, some type $t$ carries a property
  $p\in\pi_1(t)\setminus\pi_2(t)$. The pair $\left(t,p\right)$ is a feature of
  $\schema{1}$ but, as $p\notin\pi_2(t)$, not of $\schema{2}$. The case
  $\lambda_1\neq\lambda_2$ is symmetric, via $\left(t,\ell\right)$.
  \item \emph{The vocabularies and attachments agree but an endpoint differs.}
  Suppose Cases~1 and~2 do not hold and $\xi_1\neq\xi_2$ (the case
  $\tau_1\neq\tau_2$ is symmetric). Some $e\in\mathcal{E}_{S_1}=\mathcal{E}_{S_2}$ then
  has $\xi_1(e)\neq\xi_2(e)$. The triple
  $\left(e,\xi_1(e),\text{src}\right)$ is a feature of $\schema{1}$; since the
  only source feature of $e$ in $\schema{2}$ is
  $\left(e,\xi_2(e),\text{src}\right)$, it is absent from $\schema{2}$.
\end{enumerate}

The cases are exhaustive: once the vocabularies and attachments coincide, only
the endpoint maps remain, which is Case~3. Hence
$F\!\left(\schema{1}\right)\neq F\!\left(\schema{2}\right)$ whenever
$\schema{1}\neq\schema{2}$, which proves the contrapositive.
\end{proof}


\section{\graft's transformation generation and execution}\label{sec:transformations}

This validation strategy is implemented in
Algorithm~\ref{algo:transform-graph} below.

\begin{algorithm}
  \SetAlgoLined
  \DontPrintSemicolon
  \SetKwComment{tcc}{$\triangleright$}{}
  \SetKwComment{tcc*}{$\triangleright$}{}
  \KwIn{{\meditset},\mstartrel, \mnextrel, \aschema, \target}
  \KwOut{$\schemaset$: a set of schemas}
  \AlFnt{\Small}
  \caption{\texttt{Transform-Schema}\label{algo:transform-graph}}

  $\schemaset \leftarrow \emptyset$\;
  $(\startrel, \nextrel) \leftarrow 
    \mathtt{EVAL}({\meditset},\mnextrel, \mstartrel, \aschema, \target)$\;
  $\transfoset \leftarrow \mathtt{Derive-Transform}(\startrel, \nextrel)$\;

  \ForEach{$\atransfo \in \transfoset$}{
    $\aschema' \leftarrow \aschema$\;

    \ForEach{$e \in \atransfo$}{
      \lIf{$e$ is not valid with respect to $\aschema'$}{
        break
      }
      $\aschema' \leftarrow \mathtt{Apply}(e, \aschema')$\;
    }

    \lIf{$\atransfo$ is fully applied to $\aschema'$}{
      $\schemaset \leftarrow \schemaset \cup \{(\atransfo, \aschema')\}$
    }
  }

  \Return \schemaset\;
\end{algorithm}

Algorithm~\ref{algo:transform-graph} grounds the meta-transformation
to compute \texttt{Start} and \texttt{Next} (line~2)
and derives the corresponding valid transformations (line~3).
Each is then applied incrementally to the source schema
(lines~6--9), with validity checked at every step (line~7).
Only fully valid sequences produce output schemas (line~10).

\begin{theorem}[Schema Transformation Termination]
\label{thm:transform-finite}
For any finite input schemas, Algorithm~\ref{algo:transform-graph}
terminates and produces a finite set of transformation–schema pairs.
\end{theorem}

\begin{proof}
Since the input schemas are finite, the grounding of the meta-relations
\texttt{Start} and \texttt{Next} produces finitely many edit operations.
By Theorem~\ref{thm:infer-finite}, Algorithm~\ref{algo:infer-transfo}
computes a finite set of transformation sequences. Algorithm~\ref{algo:transform-graph} iterates over this finite set
(line~4) and, for each sequence, applies finitely many edit operations
(lines~6--9). Each loop therefore executes finitely many iterations.
Hence, the algorithm terminates and produces a finite output.
\end{proof}

{
\section{\graft's Optimality under pruning}\label{sec:optimality}
\newcommand{\simf}{\ensuremath{sim}\xspace}

Let \meta\transfoset be a set of meta-transformations, \aschema the source
schema, and \target the target schema. We write $\reach(\aschema)$ for the set of
schemas reachable from \aschema using \meta\transfoset, and fix a similarity
measure \simf. Let \besttarget{} denote a most-similar reachable schema,
\[
  \besttarget \in \argmax_{\schema{k}\in\reach(\aschema)} \simf(\schema{k},\target).
\]
If \simf is exact and $\target\in\reach(\aschema)$, then $\besttarget=\target$.

To help with readability, we restate Definition~\ref{def:tmonotone} and Theorem~\ref{thm:greedy-opt}.
\begin{extdef}{\ref{def:tmonotone}}[\target-monotone meta-transformation set]
A set \meta\transfoset of meta-transformations is \emph{\target-monotone} if,
for every $\schema{k}\in\reach(\aschema)$ with
$\simf(\schema{k},\target)<\simf(\besttarget,\target)$, some
meta-transformation $\meta\atransfo\in\meta\transfoset$ admits a grounding
$\atransfo$ with $\simf(\atransfo(\schema{k}),\target)>\simf(\schema{k},\target)$.
\end{extdef}

\begin{extthm}{\ref{thm:greedy-opt}}[Greedy optimality under pruning]
Let \meta\transfoset be \target-monotone. Then for every $k\ge 1$,
$\beam\ge 1$, $n_{\max}\ge 1$, and $\theta=0$, Algorithm~\ref{algo:graft}
with the greedy heuristic returns an optimal \besttarget. 
\end{extthm}

\begin{proof}
Fix an execution of Algorithm~\ref{algo:graft}. Let $h_t=\simf(\simtarget,\target)$
be the similarity of the best schema \simtarget after $t$ iterations, and let
$h_b=\simf(\besttarget,\target)$. The best schema is replaced only by a strictly
more similar one (line~14), so $h_1\le h_2\le\cdots$ is non-decreasing. Every
reachable schema has similarity at most $h_b$, hence $h_t\le h_b$.

We first show that the current best schema \simtarget always lies in the frontier
$\schemaset_{NEW}$. When \simtarget is updated (line~14), it was just generated by
an expansion and inserted into $\schemaset_{NEW}$ (line~16). It survives the
top-$k$ pruning of that expansion (lines~11--12), as it is the most similar
schema in its generating set and $k\ge 1$. Schemas already in the frontier are
never pruned. Thus \simtarget can leave the frontier only by being selected for
expansion (line~7). Suppose this happens while \simtarget is non-optimal, i.e. $h_t<h_b$. By
\target-monotonicity, expanding \simtarget produces a schema \schema{k} with
$\simf(\schema{k},\target)>h_t$. This schema becomes the new \simtarget and
re-enters $\schemaset_{NEW}$.

Assume $h_t<h_b$. Since \simtarget attains the maximum similarity $h_t$ over
$\schemaset_{NEW}$, the greedy heuristic selects it (line~7, with
$\beam\ge 1$). Expanding \simtarget produces, by \target-monotonicity, a
schema \schema{k} with $\simf(\schema{k},\target)>h_t$. This schema is new, as
otherwise $h_t$ would not be the best similarity so far. Therefore $h_{t+1}>h_t$.
Each such improvement resets the no-improvement counter (line~14). As
$n_{\max}\ge 1$, the loop does not halt while $h_t<h_b$.

By Lemma~\ref{lem:finite}, $\reach(\aschema)$ is finite, so \simf takes finitely
many values. As the loop only halts when $h_t=h_b$ and $h_t$ strictly increases while $h_t<h_b$, we reach $h_t=h_b$ after
finitely many iterations, and \simtarget is optimal. If \target is
reachable, then $h_b=\simf(\target,\target)$ and $\|\target-\simtarget\|=0\le\theta$,
so the returned transformation is exact.
\end{proof}
}

\section{\graft's Edit operation performance\label{sec:per-edit-op}}

Figure~\ref{fig:transfos_elems} reports the per-operation performance of the
edit operations on the iBench synthetic dataset.

\begin{figure}[!htbp]
    \centering
    \includegraphics[width=\linewidth]{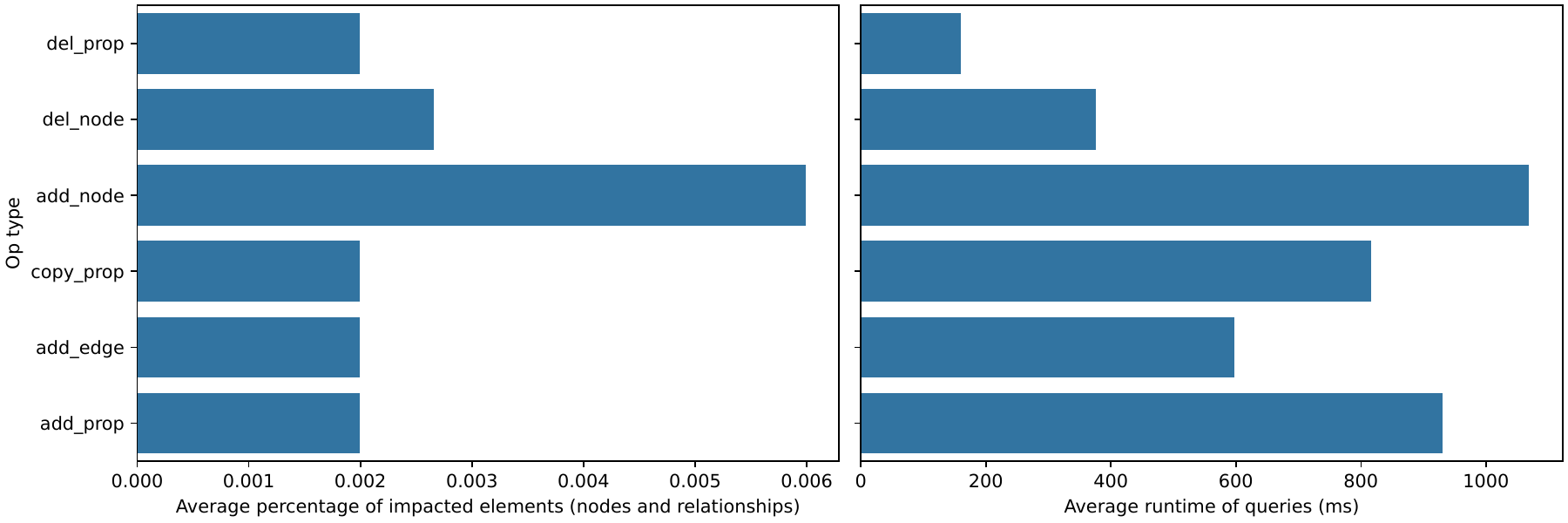}
    \caption{Average number of impacted elements per type of edit operation (left), average runtime (right).}
    \label{fig:transfos_elems}
\end{figure}

On average, each edit operation affects only a small fraction of the graph (left
plot). Node deletion has a slightly higher impact, since a deleted node may carry
an incident edge that is removed along with it, affecting both a node and an
edge. Node insertion has the largest impact, as it also adds the node's label and
an \texttt{id} property. We further distinguish \texttt{copy\_prop}, which copies
a property from one node to another, from \texttt{add\_prop}, which adds a new
property and so loads fresh data into the database.

The right plot shows the average runtime of each edit operation once translated into Cypher queries. The three deletion operations are the fastest, with node deletion
slower than property deletion because it must also remove the node's incident
edges. The remaining operations are slower still, as each matches an element by
id: a property to its host node, an edge to its endpoints, or a new node to the
node(s) it is split from or merged with.

These results show that the cost of an edit operation cannot be predicted
statically, though coarse preferences can still be set. For example, the cost of deleting a node varies not only with the number of elements of its type, but
also with the operations applied before it.
Deleting the first endpoint of an
edge costs more than deleting the second, since the first deletion already
removes the shared edge. Estimating the a priori cost of a transformation, or of
a sequence of transformations, is thus a research question of its own. The
framework can nonetheless support cost-aware search through a weighted strategy
such as \texttt{weighted\_distance}, which favors fewer transformations.

\section{\graft's Scalability}\label{sec:scalability}
\rquestapp\label{rq:scalability} \emph{Which strategy offers the best trade-off between similarity
to the target schema and runtime as the schema size grows?}

To assess scalability, we augment both schemas with identical nodes and apply
all strategies. Figure~\ref{fig:scalability-experiment}b reports
normalized similarity improvement and Figure~\ref{fig:scalability-experiment}c, the corresponding runtimes.

For most datasets, both \texttt{greedy} and \texttt{naive} strategies achieve similarity
scores that are optimal or near-optimal, where near-optimal denotes solutions
whose similarity differs from the best observed value by less than a small
margin. However, the runtime of \texttt{naive} increases sharply for
\texttt{Person-Data} and \texttt{Flight-Hotel}, eventually resulting in
timeouts. In contrast, the \texttt{greedy} heuristic consistently maintains
substantially lower runtimes while preserving high-quality solutions as schema
size increases.

The \texttt{Flight-Hotel} dataset exhibits less regular behavior: as the schema
size increases, even the \texttt{greedy} strategy may fail to reach optimal 
similarity. This effect is due to the structure of the available
meta-transformations, each of which applies exactly one edit operation. As a
result, successive candidate schemas differ by a single edit, yielding a
fine-grained search space in which many intermediate schemas offer small
similarity improvements while preventing access to larger, non-monotonic gains.

The \texttt{greedy} strategy may repeatedly select locally improving
schemas that do not lie on a globally optimal transformation path, thereby
becoming trapped in local optima. This phenomenon explains its reduced
effectiveness on the \texttt{Flight-Hotel} dataset.

Figure~\ref{fig:scalability-experiment}a plots normalized similarity improvement
against $\theta$. As expected, higher tolerance trades runtime for quality. Plateaus arise when $\theta$ exceeds the gap to the lowest
similarity schema $\schema{k}$ with $||\target - \schema{k}|| < \theta.$ For
\texttt{dblp\_to\_amalgam1}, the source schema is already 63\% similar to the
target, so the algorithm stops if $\theta=0.5.$

\begin{keytakeaway}[RQA]
\texttt{greedy} has the best similarity--runtime trade-off, matching \texttt{naive} without its timeouts, though highly granular meta-transformations (e.g., \texttt{Flight-Hotel}) may cause local optima.
\end{keytakeaway}

\begin{figure*}[!htbp]
  \centering
  \includegraphics[width=\textwidth, trim=0 .35cm 0 .32cm, clip]{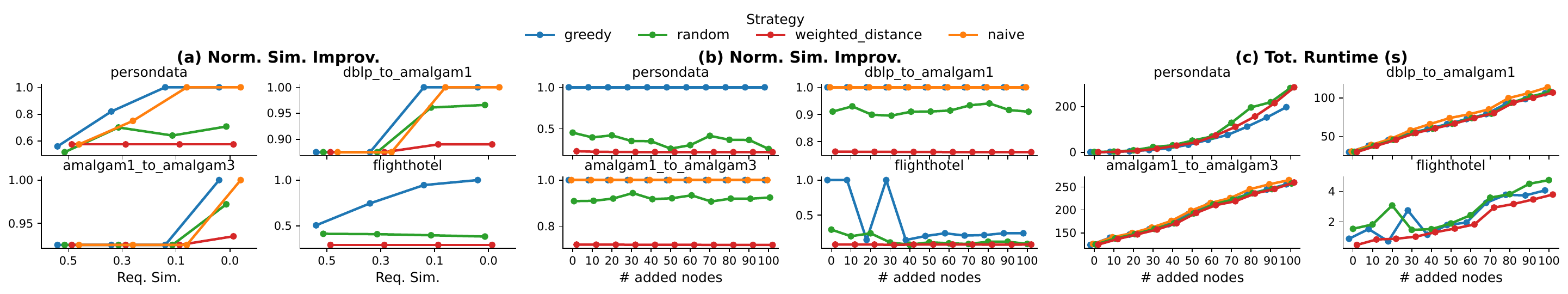}
  \caption{Y-axis: normalized similarity improvement (left and middle) and runtime; X-axis: values of $\theta$ (left), number of inserted nodes (middle and right).}
  \label{fig:scalability-experiment}
\end{figure*}

\section{Transforming the ICIJ database}\label{sec:icij-instance}

We supplement the study of \textbf{RQ\ref{rq:instance}} with the
real-world ICIJ Paradise Papers dataset, which has a smaller schema (5 nodes
and 5 edges) and a smaller database instance ($2.02$M nodes), but a more skewed
distribution across types (more than 78\% of nodes are shared between two
types). The transformations are derived from the rules in~\cite{icijdb}. For the
sake of readability, we group transformations by rule and sum similarity
improvement and measurements.

\newcommand{\comp}{{\circ}\allowbreak}
Table~\ref{tab:quality-meta-transformations} summarizes the ICIJ meta-transformations. Without idempotency, rules 1, 2, 14 and 15 are used for a resulting similarity of 73\% to the target. With idempotency, \graft uses rules 1, 6, 9, 10 and 14 which are more complex and powerful, leading to 100\% similarity.
\begin{table}[t]
\caption{ICIJ meta-transformations.}
\label{tab:quality-meta-transformations}
\centering
\small
\begin{tabularx}{\linewidth}{@{}l>{\raggedright\arraybackslash}X@{}}
\toprule
\textbf{ID} & \textbf{Meta-transformation} \\
\midrule

1  &
$\AddEdgesf(\!e,n_1,n_2\!)\comp\AddPropsf(\!n_2,p_2\!){\circ}
 \AddObjectsf(\!n_2\!){\circ}\AddPropsf(\!n_1,p_1\!)\comp\AddObjectsf(\!n_1\!)
 $
 \\[0pt]

2  &
$\AddPropsf(\!n,p\!)\comp\AddObjectsf(\!n\!)$ \\[0pt]

6  & $\AddEdgesf(\!e_2,n_2,n_3\!)\comp\AddEdgesf(\!e_1,n_1,n_2\!)\comp\AddObjectsf(\!n_3\!)\comp\AddPropsf(\!n_2,p_1\!)\comp\AddObjectsf(\!n_2\!)\comp\AddObjectsf(\!n_1\!)$ \\[0pt]

9  & $\AddPropsf(\!n_3,p_2\!)\comp\AddEdgesf(\!e_2,n_2,n_3\!)\comp\AddEdgesf(\!e_1,n_1,n_2\!)\comp\AddObjectsf(\!n_3\!)\comp\AddPropsf(\!n_2,p_1\!)\comp\AddObjectsf(\!n_2\!)\comp\AddObjectsf(\!n_1\!)$ \\[0pt]

10  &%
$\AddPropsf(\!n_1,p_k\!){\circ}\dots{\circ}\AddPropsf(\!n_1,p_1\!){\circ}\AddObjectsf(\!n_1\!)$ \\[0pt]

14 &
$\AddEdgesf(\!e,n,n\!)\comp\AddObjectsf(\!n\!)$ \\[0pt]

15 &
$\AddEdgesf(\!e,n_1,n_2\!)$ \\

\bottomrule
\end{tabularx}
\end{table}

Figure~\ref{fig:icij-instance}
shows
that the highest increase in similarity among transformations is obtained by R10,
which is also the longest runtime. However, because of the disparities between
number of nodes per type, R9 has the second longest runtime while
ranking only fourth in similarity increase. Runtime is higher than on the iBench
large-scale schema, but remains within two minutes.

\begin{figure}[!htbp]
    \centering
    \includegraphics[width=.47\textwidth, trim=0 .0cm 0 .0cm, clip]{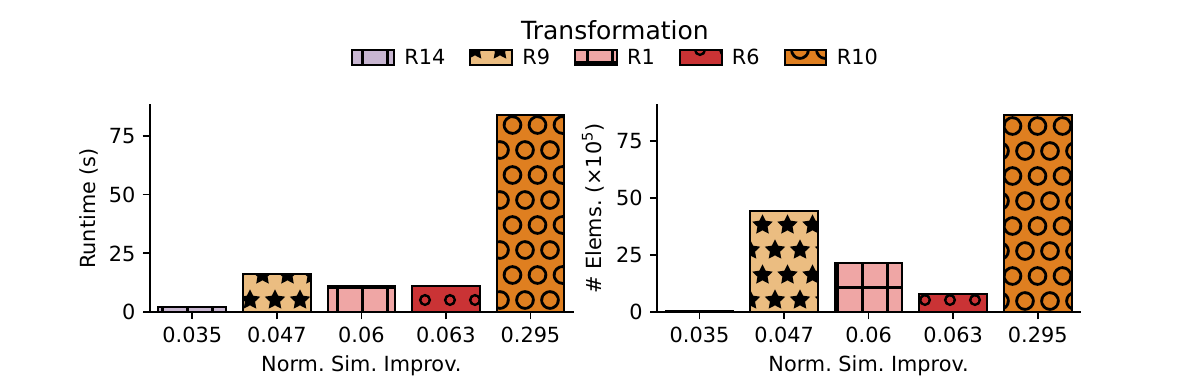}
    \caption{{Evaluation of \graft's output on ICIJ\label{fig:icij-instance}.}}
\end{figure}

\section{\graft's Runtime Breakdown}\label{sec:breakdown}

\rquestapp\label{rq:breakdown} \emph{How is the runtime of GRAFT distributed across components?}

We instrument each pipeline stage and report a runtime breakdown for each dataset using the \texttt{greedy} heuristic (pruning width 8, Jaccard similarity). This high-cost setting stresses all components. Times in Figure~\ref{fig:timings} are cumulative until termination.

We distinguish five steps:
(i) \emph{Edit operations generation} applies the rules to derive concrete edit
operations and their \nextrel/\startrel relationships from
meta-transformations (line~2 in Algorithm~\ref{algo:transform-graph});
(ii) \emph{Transformation generation} processes these tuples and contracts
cliques to produce candidate transformations (lines~1--11 in
Algorithm~\ref{algo:infer-transfo});
(iii) \emph{Schema generation} applies transformations to obtain new schemas
(lines~4--10 in Algorithm~\ref{algo:transform-graph});
(iv) \emph{Similarity computation} evaluates the similarity of each candidate
schema to the target for pruning (line~14 in
Algorithm~\ref{algo:graft}); and
(v) \emph{meta-graph construction} inserts non-pruned schemas into the meta-graph
using signatures to avoid duplicates (lines~15--17 in
Algorithm~\ref{algo:graft}). This last step is dominated by database I/O.

\begin{figure}[!htbp]
    \centering
    \includegraphics[width=.47\textwidth, trim=0 .35cm 0 .32cm, clip]{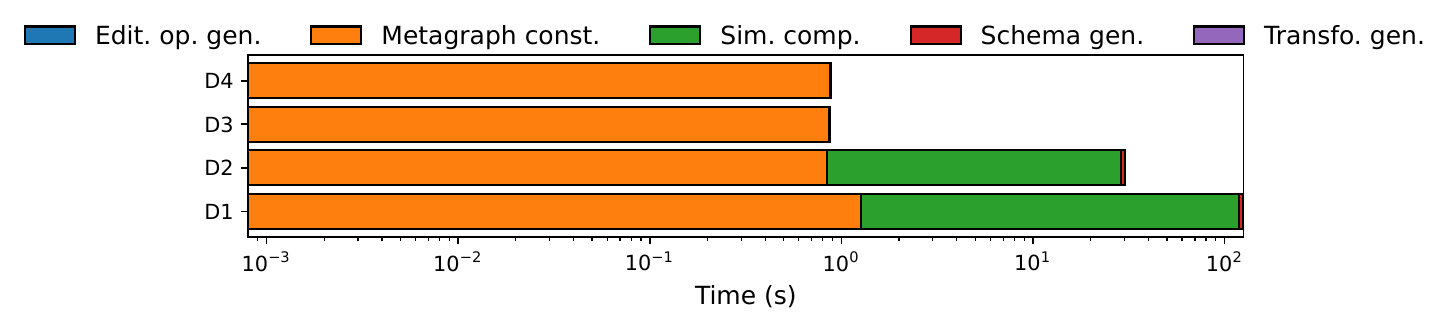}
    \caption{\graft runtime breakdown on Table~\ref{table:datasets} datasets\label{fig:timings}.}
\end{figure}

As shown in Figure~\ref{fig:timings}, similarity computation dominates the
runtime, followed by schema generation and meta-graph construction; edit
operation and transformation generation are negligible. For smaller schemas, the I/O overhead from the database overtakes the other timings which are exponentially smaller on those schemas. This is expected:
meta-transformations generate a relatively small number of edit operations,
whereas each new schema is obtained by applying a \emph{sequence} of edits,
leading to a much larger number of candidate schemas. Each such schema must be
materialized and compared to the target, making similarity the main bottleneck.
Pruning drastically reduces the number of schemas that reach the meta-graph,
keeping its construction overhead manageable. Deduplication is applied to further reduce its size.
Figure~\ref{fig:dataset-experiment}b shows the percentage of duplicated generated schemas that were filtered out.

\begin{keytakeaway}[RQB]
\graft's runtime is dominated by schema materialization and similarity
computation; edit generation is negligible.
\end{keytakeaway}

\section{Completeness of unpruned \graft}\label{sec:completeness}

\begin{theorem}[Completeness of \graft without Pruning]\label{thm:graft-complete}
Let \aschema and \target be schemas and \meta{\transfoset} a finite set of meta-transformations whose
variables range over elements of $\aschema \cup \target$. Without pruning, i.e., $k = \infty$ and
$\theta = 0$, Algorithm~\ref{algo:graft} is complete with respect to target
schema reachability (Definition~\ref{def:reachability}):  if \target is reachable from \aschema using
    \meta{\transfoset}, then
    Algorithm~\ref{algo:graft} returns a valid transformation from
    \aschema to \target.
\end{theorem}
\begin{proof}
    
In \graft, each schema reachable from the
    source is generated at most once and it stops when the threshold is reached. 
Hence, without pruning, if a path from $S$ to $T$ exists, $T$ will eventually be
generated during exploration. Then, the best approximation
maintained by the algorithm satisfies $T' = T$, and therefore $\lVert T - T'
\rVert = 0 \le \theta$ (with $\theta = 0$). The loop stops and the
algorithm terminates, returning a valid transformation from $S$ to $T$.   
\end{proof}

\begin{figure}[!htbp]
    \centering
    \includegraphics[width=.47\textwidth, trim=0 .0cm 0 .0cm, clip]{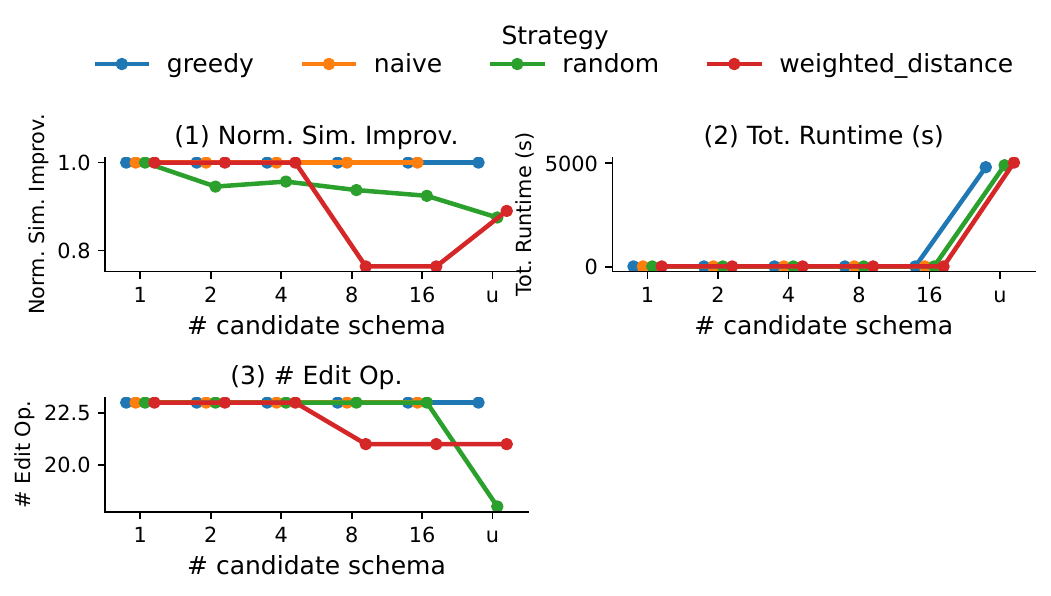}
    \caption{Performance on \texttt{DBLP-Amalgam1} with pruning\label{fig:dblp-pruning}.}
\end{figure}

\section{Performance of \graft's heuristics with pruning on \texttt{DBLP-Amalgam1}}\label{sec:dblp}

We compare the performance of the different heuristics with pruning on the \texttt{DBLP-Amalgam1} dataset. Figure~\ref{fig:dblp-pruning}
reports: (a) the normalized similarity improvement, (b) the runtime, (c) and the number of edit
operations for each heuristic. Similar to the results obtained for the \texttt{Amalgam1-Amalgam3} dataset,
\texttt{greedy} performs best.

\section{Technical Details}
\label{sec:appendix-technical}

This appendix complements Section~4.
Section~\ref{ssec:meta-trans-gen-details} details meta-transformation
generation and complexity.
Section~\ref{ssec:meta-graph-details} establishes the size bounds of the induced meta-graph.
Section~\ref{ssec:graft-details} provides details concerning the formal properties of GRAFT.

\subsection{Meta-Transformation Generation}\label{ssec:meta-trans-gen-details}

\paragraph{Clique Detection.}

Algorithm~\ref{algo:bronkerbosch} presents a basic version of the
Bron--Kerbosch algorithm~\cite{bron1973algorithm}. In our framework,
maximal clique detection is used as a preprocessing step in
Algorithm~\ref{algo:infer-transfo} to collapse permutation-equivalent
edit operations into contracted edit operations, and as a postprocessing
step to expand representative sequences when necessary. We construct a compatibility graph whose vertices are grounded edit
operations and whose edges connect mutually compatible operations,
i.e., operations $e$ and $e'$ such that each may follow the other and
share the same predecessor context. Maximal cliques correspond to
sets of edit operations that can be applied in arbitrary order.

\begin{algorithm}[ht]
\SetAlgoLined
\DontPrintSemicolon
\small
\KwIn{$P$: set of grounded edit operations, $R$: current clique (initially $\emptyset$), $X$: excluded operations (initially $\emptyset$)}
\KwOut{Maximal cliques of mutually compatible edit operations}
\caption{\texttt{Bron--Kerbosch}}
\label{algo:bronkerbosch}

\If{$P = \emptyset$ and $X = \emptyset$}{
    output $R$\;
}
\ForEach{$e \in P$}{
    $P_e \leftarrow \{ e' \in P \cap \mathrm{Next}(e) \mid
    \mathrm{pred}(e) = \mathrm{pred}(e') \text{ and } e \in \mathrm{Next}(e') \}$\;
    \texttt{Bron--Kerbosch}$(P_e, R \cup \{e\}, X \cap P_e)$\;
    $P \leftarrow P \setminus \{e\}$\;
    $X \leftarrow X \cup \{e\}$\;
}
\end{algorithm}

\paragraph{Worst-Case Complexity.}

Let $D$ denote the length of the longest path in the edit-operation
automaton, $\Delta$ the maximum number of successors of any edit
operation, and $k$ the maximum number of edit operations contracted
into a single contracted edit operation.

\begin{theorem}
Let $D$ denote the maximum depth of the DAG of edit-operations,
$\Delta$ the maximum number of successors of any edit operation,
and $k$ the maximum size of a contracted clique.
Algorithm~\ref{algo:bronkerbosch} runs in worst-case time
$O\!\left(2^k \cdot \frac{(2^k \Delta)^D - 1}{2^k \Delta - 1}\right)$.
\end{theorem}

\begin{proof}
Assume first that no edit operations are contracted and that the
successor relation forms a rooted directed acyclic graph with a
single starting operation $r$. Partition edit operations into layers
$E_0,\dots,E_D$ according to their distance from $r$.

In the worst case, every operation in $E_i$ precedes every operation
in $E_{i+1}$, with $|E_{i+1}| \le \Delta$. An operation at depth $i$
may then be explored via at most $\Delta^{i-1}$ distinct paths.
The total number of explored states is therefore bounded by
$1 + \Delta \frac{\Delta^D - 1}{\Delta - 1}$.

Consider now contraction. A contracted edit operation representing
a clique of size $k$ corresponds to $2^k$ subsets of underlying
operations. These subsets can be enumerated in $O(2^k)$ time using
incremental subset generation.

In the worst case, every operation belongs to a clique of size at most
$k$, so each layer may contain up to $2^k \Delta$ effective successors.
Substituting $\Delta$ with $2^k \Delta$ yields the stated bound.
\end{proof}

\begin{extthm}{\ref{thm:infer-sound}}[Exhaustive Duplication-free Meta-transformation Derivation]
Algorithm~\ref{algo:infer-transfo} enumerates the exact
duplication-free transformation sequences that start with an edit
operation in $\startrel$ and respect the $\nextrel$ relation.
\end{extthm}

\begin{proof}
  Generation increases transformation length iteratively. Let $\transfoset_{t}$ denote the set of generated transformations with $t\ge 1$ edit operations. The set $\transfoset_1$ is obtained at line 2 and contains only transformations composed of one edit operation which is allowed by the \startrel rules by construction.

  Assume $t > 1$ and consider the set $\transfoset_{t}$. This set is generated
after $t-1$ iterations of the loop at line 3. Let
$\atransfo=\{e_t \circ e_{t-1} \circ \dots \circ e_0\} \in \transfoset_t$.

\atransfo is obtained from a transformation \transfo{t-1} in $\transfoset_{t-1}$
by the instruction at the line 7. By induction on $t$, $e_0 \in \startrel.$ The loop at
line 6 restricts $e_t$ to be in $\nextrel(e_{t-1})$, therefore ensuring \nextrel
is respected. Lastly, line 7 ensures $e_t$ is not an edit operation of
\transfo{t-1} and thus that \atransfo does not contain duplicated edit
operations.

For each $t$, if $\transfoset_t$ is not empty, the loop at line 3 iterates and
$\transfoset_t$ is added to the result. If $\transfoset_t$ is empty, then for
each transformation
$\transfo{t-1}=\{e_{t-1} \circ \dots \circ e_0\} \in \transfoset_{t-1}$, either
no edit operation is allowed to follow $e_{t-1}$ by \nextrel
($\nextrel(e_{t-1}) = \emptyset$) or all candidate edit operations are already
in $\transfo{t-1}$
($\forall e \in \nextrel(e_{t-1}), e \in \{e_{t-1}, \dots, e_0\}$). By lines 6
and 7, the algorithm will then indeed produce an empty $\transfoset_t$. In this
case, there is also no transformation of length greater than $t+1$ and stopping
computation from the condition in line 3 will not miss any transformation.

Therefore, all the transformations are produced and they all respect the stated constraints.\qedhere
\end{proof}

\begin{theorem}[Worst-Case Time Complexity of Meta-Transformation Generation]\label{thm:infer-cplt}
Let $n$ be the number of distinct edit operations.
Then Algorithm~1 runs in $O(n!)$
in the worst case.
With clique compression,
this becomes $O(2^n)$.
\end{theorem}

\begin{proof}
Notice first that Algorithm~\ref{algo:infer-transfo} generates each transformation in $O(1).$ Its worst-case time complexity is thus equal to the worst-case number of transformations.

Let $n$ denote the number of distinct edit operations. In the worst case,
$|\startrel| = n$ and, for every edit operation $e$, $|\nextrel(e)| = n$.
Since duplicate edits are disallowed, a transformation of length $t$
can be extended by at most $n - t$ operations.

Hence, the number of transformations of length $t$ is bounded by the
number of permutations of $t$ elements drawn from $n$, i.e., $P(n,t) = \frac{n!}{(n-t)!}$. Summing over all lengths produces $\sum_{t=1}^{n} P(n,t) \in O(n! \cdot n)$. This worst case corresponds to a degenerate \nextrel\ relation in which every edit operation may follow any other. When a subset of edit
operations is pairwise compatible and commutative\footnote{\emph{e.g.,}
adding multiple properties to a node.}, we treat it as a \emph{clique}.
Such a clique can be compressed into a symbolic edit operation during
enumeration. Since permutations of the same subset are equivalent under
commutativity, it suffices to enumerate subsets rather than permutations,
reducing the worst-case number of transformations to $2^n$.
\qedhere
\end{proof}

\begin{theorem}[Worst-Case Time Complexity of Transformation Generation]
Let $n$ be the number of distinct grounded edit operations.
Algorithm~\ref{algo:transform-graph} runs in
$O(2^{n}\cdot n)$ time in the worst case (with clique compression).
\end{theorem}

\begin{proof}
Let $n$ be the number of distinct grounded edit operations
produced at line~2.
Since meta-edit variables range only over elements of the finite input schemas,
$n$ is finite and polynomial in the input size. By Theorem~\ref{thm:infer-cplt}, with clique compression,
Algorithm~\ref{algo:infer-transfo} (line~3) generates at most $2^n$
duplication-free transformations.
Each transformation has length at most $n$. Algorithm~\ref{algo:transform-graph} iterates once per transformation
(line~4) and applies at most $n$ edit operations per transformation
(line~6), each in polynomial time.

Therefore, the overall worst-case running time is
$O(2^n \cdot n)$, which is exponential in the number of grounded edit operations.
\end{proof}

\subsection{Meta-Graph Schema Search Space}\label{ssec:meta-graph-details}

\begin{theorem}[Meta-Graph Size Bound]\label{thm:meta-graph-size}
Let $n$ be the total number of distinct edit operations generated by the meta-transformations in $\meta{\transfoset}$.
Then, the number of distinct schemas
reachable under a finite set of meta-transformations
is bound by $2^{n}$.
Consequently, the meta-graph contains at most
$2^{n}$ vertices and at most
$|\meta{\transfoset}| \cdot 2^{n}$ edges.
\end{theorem}

\begin{proof}
Let \aschema be a schema in the meta-graph and let $n$ be the total number of
distinct edit operations generated by the meta-transformations. This schema is
either one of the input schemas or obtained from one of them using a
transformation. These transformations, by construction, cannot insert elements
that are not present in one of the input schemas.

Let $\editset$ be the set of all the $n$ edit operations produced by the
meta-transformations, and let $\schema{i}$ be a schema in the meta-graph. There
is a sequence of unique edit operations that transform the source \aschema into
\schema{i}. Indeed, consider the shortest path in the meta-graph from \aschema to
\schema{i} and let \atransfo be the composition of these transformations.
\atransfo is a transformation too since it is a composition of edit operations. It contains at most $n$ distinct edit operations and transforms \aschema into \schema{i}.

Thus, the total number of schemas in the meta-graph is bounded by the total
number of sequences of edit operations, which is $O(n!).$ However, every
permutation of the same edit operations leads either to the same schema (the
sets of objects added and removed are the same) or is an invalid sequence (e.g.,
setting a node property before adding it). Therefore, the number of schemas in the meta-graph is bounded by the number of sets of edit operations, that is $O(2^n).$ A given transformation applied to a schema can only produce one schema. The meta-graph can thus have at most $2^n$ edges per meta-transformation, that is, at most $|\meta{\transfoset}| \cdot 2^{n}$.
\end{proof}

\subsection{GRAFT Theoretical Properties}\label{ssec:graft-details}

\begin{extthm}{\ref{thm:graft-termination}}[GRAFT Termination]
Let $\mathcal{S}$ be a finite set of schemas and $\meta{\transfoset}$ a finite set of 
meta-transformations whose variables range over elements of 
$\mathcal{S} \cup \{T\}$. Then Algorithm~\ref{algo:graft} terminates.
\end{extthm}

\begin{proof}
We prove termination by showing that only finitely many
distinct schemas can be generated.

By Theorem~\ref{thm:infer-finite}, grounding the meta-transformations
over the input schemas yields a finite set $\transfoset$ of
duplication-free transformations. By construction, transformations
do not contain duplicate edit operations
(Theorem~\ref{thm:infer-sound}), and each edit operation refers only
to schema elements occurring in the input schemas.
Hence, no transformation can introduce fresh symbols.

Let $T$ be the finite set of schema elements (nodes, edges,
labels, properties) occurring in the input schemas.
Every schema generated during the execution of the algorithm
is obtained by applying a finite sequence of grounded
transformations to the source schema.
Since transformations cannot introduce elements outside $T$,
every generated schema is a subset of $T$.
Therefore, by Theorem~\ref{thm:meta-graph-size},
the number of distinct reachable schemas is finite.

The algorithm maintains a set $G$ containing all previously
generated schemas. No schema is ever removed from $G$,
and a schema is added to $\schemaset_{NEW}$ only if it is not
already contained in $G$. Hence, each schema can be inserted
into $\schemaset_{NEW}$ at most once.

Since the total number of distinct reachable schemas is finite,
$\schemaset_{NEW}$ can grow only finitely many times.
Once all schemas that can be reached have been added to $G$,
no new schema can be generated, and
$\schemaset_{NEW}$ strictly decreases at each iteration.
Therefore, the while-loop at line~8 terminates.

All remaining loops iterate over finite sets
($\meta{\transfoset}$ and $\schemaset'$),
and thus also terminate. Hence, the algorithm terminates.
\end{proof}

\begin{theorem}[\graft Time Complexity]\label{thm:grapft-cplxt}
Let $d$ denote the minimum number of transformation steps required to map \aschema to \target, $b$ the branching factor induced by grounded meta-transformations, and $C$ the cost of executing Algorithm~\ref{algo:transform-graph}. Without pruning, Algorithm~\ref{algo:graft} explores $O(b^d)$ schemas in the worst case, with runtime $O(b^d \cdot C)$. With pruning and width $k$, at most $k^i$ schemas are retained at depth $i$, leading to the runtime bound $O(k^d \cdot b \cdot (C + \log k))$.
\end{theorem}

\begin{proof}
We model the execution of Algorithm~\ref{algo:graft} as a search tree of depth $d$, where each level corresponds to one transformation.

Without pruning, at each step, applying the grounded meta-transformations to a schema produces at most $b$ successor schemas. In the worst case, the full search tree up to depth $d$ is explored, generating $\sum_{i=1}^{d} b^i = O(b^d)$ distinct schemas. For each generated schema, Algorithm~\ref{algo:transform-graph} is executed once, at cost $C$. Hence, the overall runtime is $O(b^d \cdot C)$.

When pruning is enabled, at depth $i$ the algorithm retains at most $k^i$ schemas. Expanding these schemas yields at most $k^i \cdot b$ candidates. For each candidate, Algorithm~\ref{algo:transform-graph} is run at cost $C$, while keeping the top-$k$ newly produced schemas by similarity to \target. This selection can be implemented using a size-bounded heap, with $O(\log k)$ update cost per candidate. Summing over all depths gives a total cost of $\sum_{i=0}^{d-1} k^i \cdot b \cdot (C + \log k)$ and bound $O(k^d \cdot b \cdot (C + \log k))$.
\end{proof}

\end{document}